\documentclass[prxquantum,reprint, superscriptaddress]{revtex4-1}
\usepackage{amsmath,amscd,amsthm,amssymb}
\usepackage{mathrsfs}
\usepackage{graphicx}
\usepackage{epsf,epstopdf}
\usepackage{caption3}
\usepackage[all]{xy}
\usepackage{tikz}
\usepackage{enumerate}
\usepackage{mathtools}
\usepackage[unicode=true,pdfusetitle, bookmarks=true,bookmarksnumbered=false,bookmarksopen=false, breaklinks=false,pdfborder={0 0 0},backref=false,colorlinks=false] {hyperref}
\hypersetup{colorlinks,linkcolor=myurlcolor,citecolor=myurlcolor,urlcolor=myurlcolor}
\usepackage{graphics, times, braket, color, bm}
\usepackage[up]{subfigure}
\usepackage{cleveref}
\usepackage{multirow}
\usepackage{array}
\usepackage{xcolor}

\definecolor{myurlcolor}{rgb}{0,0,0.7}

\def\be{\begin{equation}}
\def\ee{\end{equation}}
\def\bea{\begin{eqnarray*}}
\def\eea{\end{eqnarray*}}

\theoremstyle{plain}
\newtheorem{thrm}{\protect\theoremname}
\newtheorem{defi}[thrm]{Definition}

\newtheorem{rema}[thrm]{Remark}
\providecommand{\theoremname}{Theorem}

\newcommand*{\myproofname}{Proof}

\makeatother

\def\cD{\mathcal{D}}
\def\cF{\mathcal{F}}\def\cH{\mathcal{H}}

\def\cR{\mathcal{R}}

\def\rD{\mathrm{D}}

\def\rM{\mathrm{M}}
\def\rR{\mathrm{R}}

\theoremstyle{definition}

\theoremstyle{remark}

\graphicspath{{figs/}}

\begin{document}

 \author{Sunho Kim}
 \email{kimsunho81@hrbeu.edu.cn}
 \affiliation{School of Mathematical Sciences, Harbin Engineering University, Harbin 150001, China}

  \author{Chunhe Xiong}
 \email{Corresponding author: xiongchunhe@csu.edu.cn}
\affiliation{School of Mathematics and Statistics, HNP-LAMA, Central South University, Changsha 410083, China}

 \author{Junde Wu}
 \email{Corresponding author: wjd@zju.edu.cn}
 \affiliation{School of Information Engineering, Minnan University of Sciences and Technology, Quanzhou 362700, PR~China}
\affiliation{School of Mathematical Sciences, Zhejiang University, Hangzhou 310027, PR~China}

\title{Attainable Boundaries versus Operational Equivalence: Dual Benchmarks for the Hierarchical Assessment of Quantum Resources}
\begin{abstract}
While conventional resource theories quantify utility relative to an idealized free set, practical implementations are confined to a
specific supply environment $\mathcal{E}$ that encompasses all states realizable under current experimental and noise constraints, beyond
mere free states.  We propose a framework for evaluating the relative utility of arbitrary states against $\mathcal{E}$, focusing on two
representations: the set reachable by free operations, defining the attainable boundary, and the set from numerical optimization,
defining the criterion of operational equivalence.
Our framework establishes two core evaluation methodologies. First, by assessing a state's performance relative to the reference set under
its individually optimized strategy, we provide a strategy-dependent benchmark, thereby quantifying the achievable performance. Second, by
considering the optimal strategy for each state individually, we derive a computable lower bound on this superiority, enabling the
quantification of strategy-independent relative performance irrespective of specific experimental setups.

To translate these two methodologies into concrete numerical quantities, we introduce relative measures based on the generalized
robustness measure and the relative-entropy measure.  Building on this, we theoretically establish a quantitative hierarchical analysis
under the two aforementioned strategies---namely, the strategy-dependent benchmark and the strategy-independent lower
bound.

In a thermodynamic environment, the minimum work cost required to synthesize a target state is given by the free-energy difference
between the states; however, this is merely a lower bound derived from absolute utility, without taking the distance between states
into account.  By applying the two methodologies of relative evaluation introduced above, we show that this lower bound
corresponds to a strategy-independent quantification.  Furthermore, through two scenarios that consider the information-theoretic
distance and the thermodynamic distance, respectively, we quantify the individually optimized cost---required for synthesizing the
target state from the attainable optimal states---and thereby derive upper bounds on the minimum work cost.  Based on these scenarios, we
analyze the work cost for recovering quantum entanglement in a noisy environment.  Our analysis reveals that, in the low-noise regime at
moderate and low temperatures, the upper bound based on the thermodynamic distance is tighter than that based on the information-theoretic distance.
\end{abstract}
\maketitle

\section{Introduction and Motivation}

Quantum Resource Theory (QRT) stands as the axiomatic bedrock of modern quantum information science, providing a formal framework for characterizing useful non-classical attributes---such as entanglement \cite{Horodecki2, Vedral}, coherence \cite{Baumgratz}, magic states \cite{Veitch2014}, and asymmetry \cite{GourSpekkens2008}---as quantifiable resources.
The field traces its origins to the late 1990s with the formalization of entanglement manipulation under Local Operations and Classical Communication (LOCC) \cite{Bennett, Bennett1996}, establishing the foundational principles of entanglement purification and teleportation.
Foundational works established that entanglement cannot be generated under LOCC, laying the groundwork for the concept of "free operations" \cite{Horodecki2}. The early 2000s saw pivotal developments; Oppenheim, Horodecki, and Ekert \cite{Oppenheim2002, Horodecki1} established rigorous mathematical conditions for entanglement transformations, while Schumacher \cite{Schumacher2003} first coined the term "resource theory," solidifying its identity as an independent discipline. A major leap occurred in 2008 when Gour and Spekkens \cite{GourSpekkens2008} generalized the framework beyond entanglement to quantum reference frames, revealing that any symmetry constraint could define a resource theory. This demonstrated that QRTs are universal languages for classifying useful properties under physical constraints. By the mid-2010s, the integration of coherence theory \cite{Baumgratz} and magic state theory \cite{Veitch2014} successfully incorporated elements essential for practical quantum computation into the QRT formalism.

The pinnacle of this evolution is the proof of asymptotic reversibility by Brand\~{a}o and Gour \cite{Brandao2015-1}, which established the relative entropy of resource as the unique asymptotic quantifier when resource transformations are reversible under optimal protocols with infinite copies.
This endowed QRT with mathematical completeness, a structural foundation later synthesized by Chitambar and Gour's comprehensive review \cite{Chitambar}, which standardized the roles of free states, free operations, and resource monotones. Concurrently, the operational meaning of robustness measures was clarified by Takagi et al. \cite{Takagi} and Oszmaniec and Biswas \cite{Oszmaniec2019}, linking generalized robustness to maximal advantages in subchannel discrimination tasks. Furthermore, the development of one-shot resource theory \cite{Winter, Napoli} provided a framework for quantifying resources in the non-asymptotic regime, addressing the limitations of infinite-copy assumptions.

However, practical quantum systems operate under constraints that invalidate the assumptions of asymptotic idealism. Ideal free operations are rarely permitted; instead, the physical scope of provisionable resource states is strictly limited by specific experimental environments (e.g., temperature, defect density, decoherence times) and available operational sets (e.g., specific gate libraries). We define the set $\mathcal{E}$ as the ensemble of resource states realizable under these constrained supply environments---representing the practical boundary of what is physically permissible given current conditions and operational capabilities. While recent efforts have extended QRT towards one-shot and finite-round frameworks
\cite{Winter, Napoli, Regula, Zhu2025, Spee2017, Kaur2021}, instrumentation-based theories \cite{Buscemi2020, Ji2024, Ghai2025}, thermodynamic resource theories \cite{Brandao2013, Gour},
and dynamical resource theories \cite{Gour1, Gour2, Saxena, Kim}, existing studies largely retain the absolute set of free states $\mathcal{F}$ as the benchmark. Consequently, a theoretically superior state $\rho_{\text{out}} \notin \mathcal{E}$ with a higher resource quantifier $\rR(\rho_{\text{out}})$ may incur prohibitive overhead or performance degradation when transformed into or utilized compared to a less pure but readily available state $\rho_{\text{in}} \in \mathcal{E}$. For instance, a highly pure state incompatible with the native gate set may exhibit lower practical efficiency than a sub-optimal state tailored to the current operational constraints.

Conversely, a paradoxical scenario may arise where $\rho_{\text{out}}$, despite being outside the supply boundary $\mathcal{E}$, offers superior practical utility. If $\rho_{\text{out}}$ possesses high structural affinity with a specific operational environment---defined by the implementation of key algorithmic gates, control architecture involving measurement and feedback, memory coherence times, and the available gate library---it may be utilized immediately without costly conversions. In contrast, $\rho_{\text{in}}$ might require additional operations (transformations, distillation) that accumulate performance losses. Thus, despite a lower theoretical resource quantifier, $\rho_{\text{out}}$ can achieve higher practical efficacy when balancing operational costs and final performance.

We therefore establish $\mathcal{E}$ as the new coordinate system for resource evaluation. Our core inquiry is: \textit{Among states supplied within $\mathcal{E}$ and theoretically superior states existing outside $\mathcal{E}$, which exhibits higher relative efficiency in actual task performance under current constraints?} This shifts the focus from absolute state values to practical performance under physical limitations. In this context, we propose a relative resource evaluation framework that incorporates auxiliary reference sets---specifically, those reachable via free operations and those yielding equivalent performance in targeted tasks---to formalize the Measure of Relative Superiority. Crucially, this measure preserves the core axioms of monotonicity, convexity, and faithfulness within this relative context, providing a robust tool for practical resource assessment.

This measure supports a dual benchmarking scheme contingent on the reference set configuration. First, it facilitates strategy-dependent benchmarking, discerning whether an external state $\rho_{\text{out}}$ is practically superior to internal states $\rho_{\text{in}}$ under specific operational constraints (e.g., control structures or decoherence channels). Second, it enables strategy-independent benchmarking across the entire feasible strategy space, assessing universal performance dominance unconstrained by specific environmental factors. This duality allows a single metric to bridge task-specific optimization and universal resource valuation, in the spirit of one-shot resource formation and distillation \cite{Winter, Napoli, Brandao2015-1}.

This work resolves the central challenge of Strategy-Dependent Operational Efficacy---a core issue confronting QRTs. By establishing $\mathcal{E}$ as an explicit reference benchmark, we bridge the gap between theoretical ideals and practical implementations, providing a novel standard for fair comparisons across heterogeneous quantum supply environments. This is not merely a modification of existing theories, but an essential evolution enabling QRT to transcend asymptotic elegance in favor of tangible utility. Central to this advancement is a unified framework for quantifying relative resource efficiency from multiple perspectives. Leveraging the relative perspective, we propose relative-entropy and generalized-robustness-based measures, and perform a hierarchical analysis within the two proposed benchmarking frameworks.  This analysis simultaneously provides a strategy-dependent benchmark for immediate performance and a strategy-independent lower bound for
intrinsic potential.  The practical validity of this framework is demonstrated by analyzing the minimum work cost of restoring a target
state in a thermodynamic environment, from the two perspectives of the relative entropy and the newly introduced thermodynamically anchored robustness (TAR).

\section{Quantum resource theory in relative terms}

We examine the relativity of resources among quantum states from the perspective of resource theory.  Specifically, we characterize the
relative resources between a given set of quantum devices (states, channels, measurements, etc.) in a specific noisy environment and any
devices outside that set.  To clarify this relative perspective, we consider concrete physical scenarios.

Owing to practical limitations such as cost or technological and environmental constraints, preparing arbitrary quantum devices is
often impossible or highly challenging.  By contrast, certain facilities can, at a cost, prepare devices that possess specific
quantum capabilities, albeit in a limited fashion.

Standard resource theories provide powerful absolute quantifiers such as relative entropy and the robustness measure, and the conventional
approach of anchoring monotones to a fixed set of free states $\mathcal{F}$ has been indispensable in establishing the theoretical
bedrock of the field.  Building upon this robust foundation, this work aims to broaden the analytical perspective.  Specifically, we
define $\mathcal{E}$ as the set of resource states that are preparable via currently operational protocols or are known to be
advantageous for storage and distribution.  We then posit this set $\mathcal{E}$ itself as a new reference framework.

Approaching the problem from this vantage point reveals that the resource value of an arbitrary state cannot be adequately assessed
merely as a distance from the fixed free set $\mathcal{F}$; rather, it necessitates a holistic illumination through the lens of
practical utility relative to $\mathcal{E}$.  To illustrate, consider an entangled state $\sigma$, encoded in a specific nonlocal basis,
which is imported into the laboratory but originates from an incompatible quantum environment.  This state $\sigma$ lies outside
the scope of $\mathcal{E}$ (i.e., $\sigma \notin \mathcal{E}$). Our primary interest shifts from quantifying the intrinsic worth of
$\sigma$ to evaluating its efficiency relative to the states within $\mathcal{E}$.

Furthermore, we propose the following practical scenario.  Suppose one considers the work cost of restoring a specific target state
from a given supply environment $\mathcal{E}$.  Depending on the perspective adopted, we can analyze the relative work cost with
respect to a reference set $\mathcal{R}(\mathcal{E})$, which is taken to be either the set of all states reachable from $\mathcal{E}$ via free operations or the
set of all states that possess efficiency equivalent to that of $\mathcal{E}$.  Here, the relative efficiency is characterized
through the relative gap with respect to an extension of $\mathcal{E}$.  This sheds light on resources from the perspective of
task implementation and provides a decisive tool for re-evaluating the practical worth of states traversing the boundaries of
$\mathcal{E}$, in terms of their operational efficiency within the current laboratory context.

Accordingly, this paper proposes a framework for quantifying the relative efficiency of resources against tailored reference
benchmarks.  Specifically, depending on the objective and environmental scenario, we introduce the following two $\mathcal{E}$-inclusive sets.
\begin{enumerate}
    \item \textbf{$\mathcal{R}^{FO}(\mathcal{E})$}: The set of states reachable via free operations (e.g., LOCC). This defines the limits of transformation achievable within the current operational environment.
    \item \textbf{$\mathcal{R}^{\max}(\mathcal{E})$}: The set of states that are inaccessible via free operations but exhibit equivalent performance to the optimal states in $\mathcal{E}$ for target quantum tasks (e.g., teleportation, metrology). This reflects a perspective that prioritizes practical efficacy over mathematical transformability.
\end{enumerate}
This multi-layered approach enables a nuanced valuation of resources that accounts for transformability (Reachability), practical efficacy (Performance), and environmental fitness.

The primary significance of this work lies in shifting the paradigm of quantum resource quantification from an absolute to a relative coordinate system. The framework allows for the coexistence of the traditional evaluation scheme, which prioritizes mathematical axioms such as convexity and monotonicity as in conventional monotones, and the relative evaluation scheme, which assesses resources based on operational efficiency relative to a user-defined environment ($\mathcal{E}$) and operationally or task-specific baselines ($\mathcal{R}^{FO}(\mathcal{E}), \mathcal{R}^{\max}(\mathcal{E})$). In particular, by distinguishing between the reachability via free operations ($\mathcal{R}^{FO}(\mathcal{E})$) and performance-equivalent states ($\mathcal{R}^{\max}(\mathcal{E})$), users can select or combine these two schemes depending on their research objectives or experimental environment. This provides a systematic method to identify resource values that absolute metrics fail to capture. Consequently, this work not only complements existing resource theories but also offers a versatile and decisive benchmarking tool for comparing disparate quantum platforms and optimizing resource allocation in real-world quantum information processing tasks.


\begin{figure}
\includegraphics[width=3.4in]{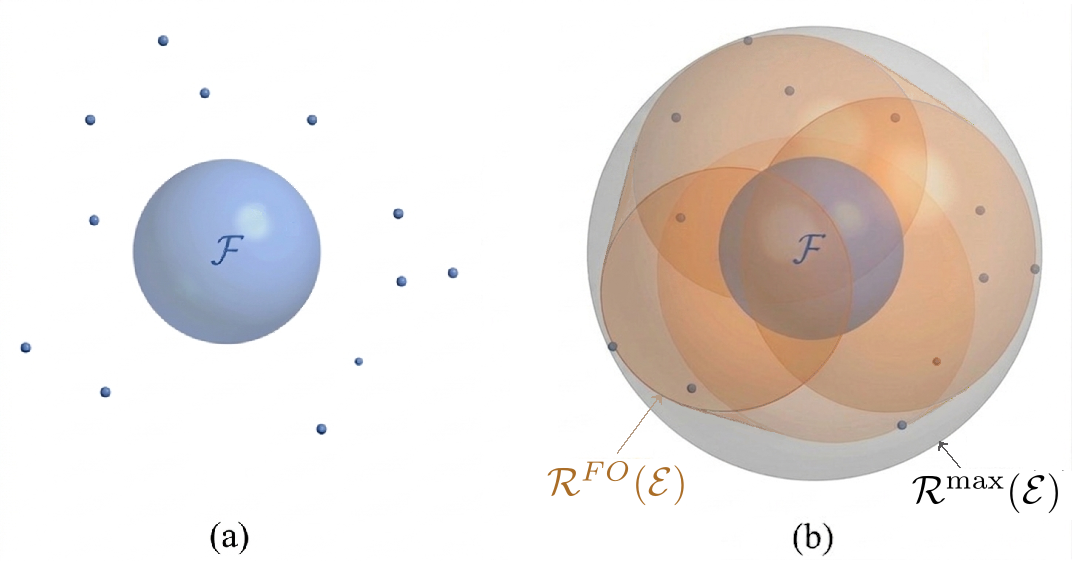}
\caption{(a) Let $\mathcal{E}$ denote the set of currently accessible or preparable states---be they finite, countably infinite, or continuous. While this set can generally be assumed to encompass all free states, no universal quantitative method is known that guarantees both convexity and monotonicity under free operations when attempting to quantify relative resources relative to $\mathcal{E}$. (b) To resolve this issue, we construct two reference sets. The first is $\cR^{FO}(\mathcal{E})$, defined as the subset of states reachable from $\mathcal{E}$ exclusively via free operations. The second is $\cR^{\max}(\mathcal{E})$, comprising all states that, although inaccessible via free operations, yield equivalent performance in the given quantum task.}
\label{fig1}
\end{figure}

Next, we turn to the specific task of constructing the reference sets $\mathcal{R}^{FO}(\mathcal{E})$ and $\mathcal{R}^{\max}(\mathcal{E})$. We propose a systematic framework for their physically motivated construction and for the subsequent evaluation of relative resource advantages.

We define $\mathcal{E}$ as the restricted set of quantum states preparable within given technological or economic constraints; by construction, this set is assumed to encompass all free states.
The central question we address is how to quantify the relative resource advantage of quantum states lying outside this currently preparable set $\mathcal{E}$.
Guided by the foundational principles of quantum resource theories, any meaningful quantification is recommended to satisfy three fundamental criteria: faithfulness, monotonicity, and convexity.
These requirements emerge naturally from physical considerations and provide mathematically robust conditions for any viable resource measure.

However, we identify a significant obstacle to this approach: if $\mathcal{E}$ lacks convexity or exhibits severe geometric distortion relative to the free set, constructing a measure that preserves monotonicity under free operations becomes inherently problematic (see (a) of FIG.\ref{fig1}).
To overcome this challenge, we develop an extension procedure for $\mathcal{E}$.
The key insight lies in carefully expanding the preparable set while rigorously maintaining:
(i) the physical plausibility of the extended set;
(ii) the preservation of essential resource-theoretic properties; and
(iii) operational relevance to quantum tasks.
We present a systematic methodology for this extension, derived from concrete physical scenarios that guide the generalization process.
This proposed extension offers a principled approach to comparing quantum resources across disparate technological platforms and preparation constraints.

As $\mathcal{R}^{FO}(\mathcal{E})$ comprises all quantum states reachable from $\mathcal{E}$ via free operations, it is formally defined as:
\be
\mathcal{R}^{FO}(\mathcal{E}) = \{ \rho \mid \rho = \sum_i \Lambda_i(\sigma_i),\; \forall \Lambda_i \in FO,\; \sigma_i \in \mathcal{E} \},
\ee
where $FO$ denotes the set of all free operational protocols (see (b) of FIG.\ref{fig1}).
By construction, this set inherently guarantees convexity.
Conceptually, it represents the ensemble of states that are in principle preparable through free operations, abstracting away from practical experimental constraints such as noise or device limitations.

The scope of $\mathcal{R}^{\max}(\mathcal{E})$ is inherently task-dependent by definition.
We formalize this set through the following reasoning.
Consider a quantum processing task characterized by a suitable resource quantifier $\rM$, chosen in accordance with established quantum resource theories.
If two quantum states $\rho$ and $\rho'$ satisfy $\rM(\rho)<\rM(\rho')$, we conventionally judge the intrinsic resource content of $\rho'$ to be dominant over that of $\rho$ for this specific task.
It is crucial to note, however, that this numerical advantage does not necessarily stem from a comparison under a single, fixed operational strategy (defined by a specific channel or measurement protocol).
Generally, the measure $\rM$ is recognized as the maximum achievable efficacy relative to free states.
This implies that, under certain strategies, the effect derived from $\rho$ might yield greater advantages than that from $\rho'$.
Quantifying resources via the maximal efficacy across all feasible strategies is a fundamental guiding principle frequently adopted in conventional resource theories; consequently, applying this same principle to the quantification of relative resources is well justified.

Following this principle, if the maximal effect achievable from any state $\rho$ outside the set $\mathcal{E}$ (considering all viable strategies) fails to reach the maximal effect attained by states within $\mathcal{E}$, then we cannot deem the resources of $\rho$ to be dominant over $\mathcal{E}$, irrespective of the specific strategy employed.
Therefore, we quantify relative resources based on the set of all quantum states whose maximal efficacy across all practicable strategies is inferior to that of the states in $\mathcal{E}$ (see (b) of FIG.~\ref{fig1}).
The extended set is formally defined as:
\be
\mathcal{R}^{\max}(\mathcal{E}) = \{\rho\mid\rM(\rho)\leq \rM(\mathcal{E})\},
\ee
where $\rM(\mathcal{E}) = \max\{\rM(\sigma)\mid\sigma \in \mathcal{E}\}$. This extended set provides an ideal foundation for constructing resource measures consistent with the fundamental criteria of resource theories, namely faithfulness, monotonicity, and convexity.
Building upon this, we present a comprehensive framework for quantifying relative resources.

Since free operations are non-resource-increasing, any state reachable via $FO$ satisfies $\rM(\rho) \leq \rM(\mathcal{E})$.
Hence, we obtain the hierarchy $\mathcal{R}^{FO}(\mathcal{E}) \subseteq \mathcal{R}^{\max}(\mathcal{E})$, indicating that reachability implies performance-equivalence.

Building upon this, we present an operational resource-theoretic framework for quantifying relative resources, utilizing extended sets derived from $\mathcal{E}$---specifically $\mathcal{R}^{FO}(\mathcal{E})$ and $\mathcal{R}^{\max}(\mathcal{E})$---as reference benchmarks.
Here, the reference set $\cR(\mathcal{E})$ is deemed equivalent to $\mathcal{E}$; it represents an adaptively extended set (e.g., $\mathcal{R}^{FO}(\mathcal{E})$ or $\mathcal{R}^{\max}(\mathcal{E})$) whose scope is reset according to the specific environment and task at hand.

We now define a relative measure, denoted $\rM_{\mathcal{E}}$, to quantify the supremacy of resources relative to $\cR(\mathcal{E})$, subject to the following axioms:

\textbf{(R1)} Faithfulness: $\rM_{\mathcal{E}}(\rho)\geq 0$, with equality $\rM_{\mathcal{E}}(\rho)=0$ if and only if $\rho \in \cR(\mathcal{E})$.

\textbf{(R2)} Monotonicity under any free operation $\Lambda$:
\be
\rM_{\mathcal{E}}(\Lambda(\rho)) \leq \rM_{\mathcal{E}}(\rho)
\ee
for any free operations $\Sigma$.

\textbf{(R3)} Convexity: For any probability distribution $\{q_i\}$ and states $\{\rho_i\}$, the following holds:
\be
\rM_{\mathcal{E}}\biggl(\sum_i q_i\rho_i\biggr) \leq \sum_i q_i\rM_{\mathcal{E}}(\rho_i).
\ee
Any function $\rM_{\mathcal{E}}$ satisfying conditions (R1-R3) is termed a Measure of Relative Superiority.








\section{Hierarchy via relative measures based on robustness and relative entropy}

Here we carry out a hierarchical analysis, based on several measures, for the relative resource theory with respect to the supplying
reference set $\mathcal{E}$ introduced by the authors.  To this end, we extend our relative measures via robustness- and
relative-entropy-based quantifiers that satisfy the following definitions:
\begin{align}
\rD(\rho) &= \min_{\sigma\in \cF(\cH)}\rD(\rho\|\sigma), \quad \textmd{and}\\
\rR(\rho) &= \min_{\tau\in \cD(\cH)}\Big\{s\geq0\Big|\frac{\rho+s\tau}{1+s}\in \cF(\cH)\Big\}.
\end{align}

In this section, we establish two distinct closed convex sets within the given supply environment $\mathcal{E}$, both representing notions of ``equivalent resources,'' to facilitate the evaluation of relative resources tailored to the subchannel discrimination task.
First, we define the Free-Operation Reachable Set $\mathcal{R}^{FO}(\mathcal{E})$. As previously discussed, this set comprises all states reachable from $\mathcal{E}$ via the allowed free operations (FO). It represents the operational boundary imposed by current physical constraints (e.g., available gate libraries) and is therefore task-independent.
Second, we introduce the Max-Robustness Equivalence Set $\mathcal{R}^{\max}_r(\mathcal{E})$ and the Max-Relative Entropy Equivalence Set $\mathcal{R}^{\max}_d(\mathcal{E})$, specific to the subchannel discrimination task at hand. These sets define the boundaries of states deemed to offer an operational advantage equivalent to $\mathcal{E}$ in this particular context, respectively, as follows:
\begin{align}
\mathcal{R}^{\max}_r(\mathcal{E}) &= \Bigl\{\, \rho \;\Big|\; \rR(\rho) \leq \rR(\mathcal{E}) \,\Bigr\}, \quad \textmd{and}\\
\mathcal{R}^{\max}_d(\mathcal{E}) &= \Bigl\{\, \rho \;\Big|\; \rD(\rho) \leq \rD(\mathcal{E}) \,\Bigr\},
\end{align}
where
$\rR(\mathcal{E}) := \max\{\,\rR(\sigma) \mid \sigma \in \mathcal{E}\,\}$
and
$\rD(\mathcal{E}) := \max\{\,\rD(\sigma) \mid \sigma \in \mathcal{E}\,\}$
are the maximal robustness and relative-entropy values attainable
within the set $\mathcal{E}$.

Based on these sets, we define the following measures to evaluate relative efficiency:
\begin{defi}\label{def1}
Given a state $\rho$, we define the robustness measures of the resource with respect to the supply set $\mathcal{E}$ as follows:
\be\label{eq1}
\rR^{FO}_{\mathcal{E}}(\rho) = \min_{\tau\in \cD(\cH)}\Bigg\{s\geq0\Big|\frac{\rho+s\tau}{1+s}\in \mathcal{R}^{FO}(\mathcal{E})\Bigg\},
\ee
and
\be\label{eq2}
\rR^{\max}_{\mathcal{E}}(\rho) = \min_{\tau\in \cD(\cH)}\Bigg\{s\geq0\Big|\frac{\rho+s\tau}{1+s}\in \mathcal{R}^{\max}_r(\mathcal{E})\Bigg\}.
\ee
\end{defi}

\begin{defi}\label{def1}
Given a state $\rho$, we define the relative entropy measures of the resource with respect to the supply set $\mathcal{E}$ as follows:
\be\label{eq3}
\rD^{FO}_{\mathcal{E}}(\rho) = \min_{\sigma\in \mathcal{R}^{FO}(\mathcal{E})}\rD(\rho\|\sigma),
\ee
and
\be\label{eq4}
\rD^{\max}_{\mathcal{E}}(\rho) = \min_{\sigma\in \mathcal{R}^{\max}_d(\mathcal{E})}\rD(\rho\|\sigma).
\ee
\end{defi}
For the compact convex set $\mathcal{R}^{FO}(\mathcal{E})$ reachable under free operations, the measures defined above in terms of robustness and relative entropy
satisfy monotonicity (R2) under free operations, in the same manner as the existing measures do. Moreover, from the definition,
faithfulness (R1) and convexity (R3) are obtained naturally.

A fundamental property of free operations (FO) is that they cannot increase the amount of resource; consequently, the application of any FO to a state within the supply environment $\mathcal{E}$ cannot yield a state that outperforms $\mathcal{E}$ in any operational task.
This guarantees the inclusion $\cF(\cH) \subseteq \mathcal{R}^{FO}(\mathcal{E}) \subseteq \mathcal{R}^{\max}(\mathcal{E})$, where $\mathcal{R}^{FO}(\mathcal{E})$ denotes the set of states reachable from $\mathcal{E}$ via FO, and $\mathcal{R}^{\max}(\mathcal{E})$ denotes the set of states offering equivalent or lesser performance in a given task.
Crucially, this hierarchical relationship between the reference sets induces a corresponding ordering between our proposed measures.
From our definitions, we directly obtain the inequality: for any $\rho \in \cD(\cH)$
\begin{equation}
\rR(\rho) \geq \rR^{FO}_{\mathcal{E}}(\rho) \geq \rR^{\max}_{\mathcal{E}}(\rho),
\end{equation}
and
\begin{equation}
\rD(\rho) \geq \rD^{FO}_{\mathcal{E}}(\rho)\geq\rD^{\max}_{\mathcal{E}}(\rho).
\end{equation}

A crucial distinction exists between the two reference sets defined in this work. While $\mathcal{R}^{FO}(\mathcal{E})$ constitutes the
set of states physically reachable from the supply environment $\mathcal{E}$ under free operations, $\mathcal{R}^{\max}(\mathcal{E})$
is defined purely numerically on the basis of each measure. Consequently, membership in $\mathcal{R}^{\max}(\mathcal{E})$ does not
guarantee physical reachability within the present constrained supply environment. Furthermore, the experimental approximation of
$\rD^{\max}_{\mathcal{E}}(\rho)$ or $\rR^{\max}_{\mathcal{E}}(\rho)$ for an arbitrary state $\rho$ may be infeasible, as it requires
characterizing the boundary of a high-dimensional convex set defined by a maximization principle. Nevertheless, the numerical
relationship established in the following theorem addresses this issue. A detailed proofs of the following two theorems are provided in Appendix~\ref{Appendix A}.

\begin{thrm}
\label{thm1} For any state $\rho$, we have that
\be\label{r-max}
\rR^{\max}_{\mathcal{E}}(\rho) = \max\Big\{\frac{\rR(\rho)-\rR(\mathcal{E})}{1+\rR(\mathcal{E})},0\Big\}.
\ee
\end{thrm}

\begin{thrm}
\label{thm2} For any state $\rho$, we have that
\be\label{d-max}
\rD^{\max}_{\mathcal{E}}(\rho) = \max\Big\{\rD(\rho)-\rD(\mathcal{E}), 0\Big\}.
\ee
\end{thrm}

By proving the quantitative connection between the established
robustness measure $\rR$ and our proposed measure
$\rR^{\max}_{\mathcal{E}}$, as well as the analogous connection
between the relative-entropy measure $\rD$ and our proposed
measure $\rD^{\max}_{\mathcal{E}}$, we enable the indirect
approximation of both $\rR^{\max}_{\mathcal{E}}(\rho)$ and
$\rD^{\max}_{\mathcal{E}}(\rho)$.  This duality allows us to achieve,
for both measures in common, the following two critical objectives.

First, by evaluating the relative superiority of $\rho$ over
$\mathcal{E}$ under the optimal strategy for $\rho$, we provide a
strategy-dependent benchmark of achievable performance for both the
robustness and the relative-entropy quantifiers.

Second, by considering the optimal strategy for each state
individually, we can provide a computable lower bound for this
superiority, thus quantifying the relative performance based on
numerically accessible optimal success probabilities.  Unlike the
first objective, this approach yields a strategy-independent
numerical bound that reflects the inherent potential of $\rho$
relative to $\mathcal{E}$.

Consequently, for any state $\rho$ and for both measures, we obtain
that
\begin{equation}
\rR^{FO}_{\mathcal{E}}(\rho) \geq \max\Big\{\frac{\rR(\rho)-\rR(\mathcal{E})}{1+\rR(\mathcal{E})},0\Big\},
\end{equation}
and
\begin{equation}
\rD^{FO}_{\mathcal{E}}(\rho) \geq \max\{\rD(\rho)-\rD(\mathcal{E}),\,0\},
\end{equation}
so that the two maximal robustness and relative-entropy values can be evaluated entirely in terms of the corresponding established
quantities.

Furthermore, the numerical relations established in Theorems~\ref{thm1} and~\ref{thm2} ensure that $\rR^{\max}_{\mathcal{E}}$ and $\rD^{\max}_{\mathcal{E}}$ naturally inherit the core properties of $\rR$ and $\rD$, respectively---namely, faithfulness, monotonicity, and convexity---thus guaranteeing that they satisfy the necessary axioms for a valid measure of superiority.

In the preceding sections, we established a formal framework of relative quantifiers for a given supply resource $\mathcal{E}$, and
examined what relations the resulting feasible bounds and operational equivalences respectively impose on the measures $\rD_{\mathcal{E}}$
and $\rR_{\mathcal{E}}$. In practice, quantum technologies, particularly those in the NISQ era, demand concrete, task-specific
validations. To bridge the gap between abstract resource theories and experimental reality, we present a hierarchical analysis of the
minimum work cost required to restore the target state in a thermodynamic noisy environment. Its central idea is interpreted by
multifaceted consideration of the thermodynamic cost incurred when departing from the supply resource $\mathcal{E}$. Such an extension
consolidates the operational relevance of our framework for next-generation quantum devices.


\section{Thermodynamically Anchored Robustness and Work Cost}
\label{sec:thermal_robustness}

In realistic quantum thermodynamic settings, the available supply $\mathcal{E}$ is defined not merely as a pristine thermal bath but as the set of all quantum states that can be physically prepared or accessed in the laboratory. This encompasses the effective thermal reservoir state $\tau_\beta = e^{-\beta H}/Z$, where $Z = \text{Tr}(e^{-\beta H})$ is the partition function, which characterizes the dominant bath temperature. It may also encompass additional thermal components, noise-corrupted entangled states, partially degraded coherences, and experimentally prepared but quality-degraded resource states.

The set of states reachable from \textit{any} state in $\mathcal{E}$ via free operations (thermal operations, energy-conserving unitaries, and ancillary heat exchanges) is denoted $\mathcal{R}^{FO}(\mathcal{E})$. Note that $\tau_\beta \in \mathcal{R}^{FO}(\mathcal{E})$ by definition, but $\mathcal{R}^{FO}(\mathcal{E})$ may be strictly larger than the set reachable from $\tau_\beta$ alone when $\mathcal{E}$ already contains non-thermal resources. We also define $\mathcal{R}^{\max}(\mathcal{E})$ as the set of all states whose free energy does not exceed the maximum free energy available in the supply environment $\mathcal{E}$, i.e., $\mathcal{R}^{\max}(\mathcal{E}) = \{\sigma \mid F(\sigma) \leq \max_{\rho' \in \mathcal{E}} F(\rho')\}$.

In classical thermodynamics, the maximum reversible work extractable from an isothermal process equals the decrease in Helmholtz free energy, $W_{\max}^{\mathrm{extract}} = -\Delta F$. More generally, the second law imposes the work bound
\begin{equation}
    W \;\ge\; F(\rho)-F(\sigma)
    \label{eq:second_law_work}
\end{equation}
for any transformation $\sigma\to\rho$ under thermal operations, with equality only in the reversible macroscopic limit \cite{Brandao2015,Horodecki2013}.

In the quantum resource-theoretic setting, this principle extends naturally: the minimal work required to synthesize a target state $\rho$ from the supply environment $\mathcal{E}$ is the smallest free-energy gap that must be overcome,
\begin{align}
    W_{\min}(\mathcal{E}\to\rho)
    &:=\;
    \min_{\sigma\in\mathcal{R}^{FO}(\mathcal{E})}
    \bigl[F(\rho)-F(\sigma)\bigr]\\
    &= \min_{\sigma\in\mathcal{R}^{\max}(\mathcal{E})}
    \bigl[F(\rho)-F(\sigma)\bigr],\notag
    \label{eq:work_min}
\end{align}
where $F(\cdot)=\langle E\rangle-k_B T S(\cdot)$ is the free energy, with $\beta = (k_B T)^{-1}$.
By definition, $W_{\min}$ is purely a measure of operational efficiency based solely on free-energy differences; it completely disregards the (analytical) reachable distance from the supply environment $\mathcal{R}^{FO}(\mathcal{E})$ to the target state $\rho$.
As a work measure, $W_{\min}$ is therefore blind to the geometric and information-theoretic distance (gap) between $\rho$ and the reachable set, and thus provides merely a fundamental lower bound mandated by the second law.

Crucially, in the single-shot regime, strict inequality generally arises between the actual work cost and $W_{\min}$ due to irreversibility and coherence-locking effects \cite{Brandao2015,Horodecki2013}.
Specifically, coherence-locking implies that the dissipation of quantum resources (such as entanglement and coherence) into equilibrium states results in an excess work cost quantified by the information-theoretic statistical distance (relative entropy).
Therefore, the practical work cost is determined by reflecting the comprehensive and analytical gaps between states.
By introducing metrics that account for these analytical gaps, we establish a tighter and physically more meaningful upper bound on the minimum work cost.
\begin{rema}
It is important to note that our objective is to incorporate the information-theoretic and thermodynamic distances into the work cost, and we do not strictly require the resulting upper bound to be tight.
Rather, this implies that the proposed upper bound can approach the actual work cost more closely than the bare minimum work cost $W_{\min}$ alone, without necessarily being a tight bound in the rigorous sense.
\end{rema}

To rigorously analyze this minimal work cost $W_{\min}$, we introduce two distinct scenarios and derive corresponding upper bounds.
Specifically, we evaluate the work cost from: (i) the state within the supplied environment $\mathcal{E}$ that minimizes the relative entropy distance to $\rho$, and (ii) the state lying on the straight line connecting $\rho$ and the Gibbs state $\tau_\beta$ that minimizes the thermal gap. Through these two scenarios, we establish two independent upper bounds on $W_{\min}$.

Because target states are customarily chosen from those possessing relatively high free energy, we restrict our attention to states satisfying $F(\rho) > F(\tau_\beta)$---that is, states whose Helmholtz free energy exceeds that of the thermal reservoir state of the environment.

Scenario 1: Relative-entropy-based bound via $\rD^{FO}_{\mathcal{E}}(\rho)$.
The first scenario is formulated based on the free-operation relative entropy measure $\rD^{FO}_{\mathcal{E}}(\rho)$.
Since the relative entropy provides the most intuitive quantifier for analyzing differences in free energy, it is physically well-motivated to first convert the initial state into the state $\sigma^* \in \mathcal{E}$ possessing the minimal relative entropy via free operations, and then assess the work cost required to reach the target state $\rho$.
Within this scenario, an upper bound on the minimum work cost is derived as stated in the following theorem.

\begin{thrm}[Work--Relative Entropy Upper Bound]
\label{thm:Relative Entropy Upper Bound}
For any $d$-dimensional target state $\rho$ satisfying $F(\rho) > F(\tau_\beta)$ and $\rD^{FO}_{\mathcal{E}}(\rho)\leq 2\left(\frac{d-1}{d}\right)^2$, then
\begin{align}
&W_{\min}(\mathcal{E}\to\rho)\leq \|H\|_\infty \sqrt{2\rD^{FO}_{\mathcal{E}}(\rho)}\notag \\  &+\frac{1}{\beta}\left[\sqrt{\frac{\rD^{FO}_{\mathcal{E}}(\rho)}{2}}\;\ln(d-1)\;+\;h_2\!\left(\sqrt{\frac{\rD^{FO}_{\mathcal{E}}(\rho)}{2}}\right)\right].
\end{align}
where $h_2(x)=-x\ln x-(1-x)\ln(1-x)$ is the binary entropy.
Hereafter, we denote this upper bound by $W_{\rD}^{\mathrm{UB}}$ for convenience.
\end{thrm}
\noindent (Appendix~\ref{Appendix B} provides the proof of the above theorem.)

Scenario 2: Upper bound via thermodynamically anchored robustness (TAR).
To quantify and analyze the thermodynamic gap, we introduce the thermodynamically anchored robustness (TAR) $\rR_{\mathcal{E}}^{\text{th}}(\rho)$, defined as the minimal value $R\in[0,1]$ for which there exists a decomposition
\begin{equation}
    \sigma = (1-R)\rho + R\tau_\beta,
    \label{eq:thermal_decomp}
\end{equation}
with $\sigma \in \mathcal{R}^{FO}(\mathcal{E})$.
This represents the minimum contamination weight of the reservoir state $\tau_\beta$ that must be removed from a reachable state in order to recover the target state $\rho$. Physically, it quantifies the amount of purification required to remove thermal contamination from the supplied environment in order to synthesize $\rho$.

\begin{rema}[Observability and Physical Interpretation]
\label{rem:observability}
$R_{\mathcal{E}}^{\text{th}}(\rho)$ is directly observable in experiment through the following protocol:
\begin{enumerate}
    \item Perform quantum state tomography on $\rho$ to obtain its eigenvalues $\{p_i\}$ and eigenvectors.
    \item Measure the environment temperature $T$ (or estimate $\beta$) to construct $\tau_\beta = e^{-\beta H}/Z$.
    \item Solve the convex feasibility problem: minimize $R \geq 0$ subject to $\sigma(\rR) = (1-\rR)\rho + \rR\tau_\beta \in \mathcal{R}^{FO}(\mathcal{E})$.
\end{enumerate}
Unlike the unanchored robustness $\rR_{\mathcal{E}}^{FO}$, the thermal anchor eliminates the arbitrary auxiliary $\tau$, making the bound free of unphysical energy-scale assumptions. Furthermore, $R_{\mathcal{E}}^{\text{th}}$ reduces to the standard robustness in the limit $\beta \to 0$ (infinite temperature), where $\tau_\beta \to \mathbb{I}/d$ becomes the maximally mixed state.
\end{rema}

While both $W_{\min}$ and $\rR_{\mathcal{E}}^{\text{th}}$ are optimization problems over $\mathcal{R}^{FO}(\mathcal{E})$, the latter reduces to a tractable one-dimensional search. From \eqref{eq:thermal_decomp}, $\rR_{\mathcal{E}}^{\text{th}}$ is obtained via a scalar bisection over $R$ subject to the constraint $\sigma(\rR)\in\mathcal{R}^{FO}(\mathcal{E})$. In contrast, $W_{\min}$ requires a full optimization over the multidimensional set $\mathcal{R}^{FO}(\mathcal{E})$, which is generally computationally intractable.

To evaluate this upper bound, we introduce the thermal energy deviation $\Delta E_\beta(\rho) = \operatorname{Tr}[\rho H] - \operatorname{Tr}[\tau_\beta H]$. This serves as a temperature-dependent, experimentally observable measure of the energetic displacement between $\rho$ and $\tau_\beta$.

The aforementioned measure $\rR_{\mathcal{E}}^{\text{th}}(\rho)$, together with $\Delta E_\beta(\rho)$, yields a rigorous and experimentally accessible upper bound on $W_{\min}$, formalized as follows:

\begin{thrm}[Work--Robustness Upper Bound]
\label{thm: Robustness Upper Bound}
For any $d$-dimensional target state $\rho$ satisfying $F(\rho) > F(\tau_\beta)$ and $\rR_{\mathcal{E}}^{\mathrm{th}}(\rho)\in \Big[0, 1-\exp{\Big(-2\big(\frac{d-1}{d}\big)^2\Big)}\Big]$,
\begin{align}
&W_{\min}(\mathcal{E}\to\rho) \le \rR_{\mathcal{E}}^{\mathrm{th}}(\rho)\,\Delta E_\beta(\rho)\notag\\
&\quad + \frac{1}{\beta}\left[f\big(\rR_{\mathcal{E}}^{\mathrm{th}}(\rho)\big)\ln(d-1) + h_2\Big\{f\big(\rR_{\mathcal{E}}^{\mathrm{th}}(\rho)\big)\Big\}\right],
\end{align}
where $h_2(x)=-x\ln x-(1-x)\ln(1-x)$ and $f(x)= \sqrt{\frac{-\ln(1-x)}{2}}$.
Hereafter, we denote this upper bound by $W_{\rR,\text{th}}^{\text{UB}}$ for convenience.
\end{thrm}
\noindent (Appendix~\ref{Appendix C} provides the proof of the above theorem.)

Therefore, we have that
\begin{equation}
W_{\min}(\mathcal{E}\to\rho)\leq \min\{W_\rD^{\text{UB}}, W_{\rR,\text{th}}^{\text{UB}}\}.
\end{equation}

\begin{rema}
Let us consider the two aforementioned scenarios under the assumption
that the respective informational and thermodynamic distances cannot be
sufficiently contracted via free operations. In such a case, if one
attempts to forcibly pursue the minimum work cost, the imposed work
process merely results in the dissipation of the system's free energy
into heat. This leads to a thermodynamically inefficient and
unphysical regime in which the distances themselves remain insufficiently
reduced while free energy is exclusively expended. Consequently, such a
regime is incompatible with the scenario of thermal robustness and the
TAR under consideration.

From this perspective, although the bounds imposed on $\rD^{FO}_{\mathcal{E}}$ and $\rR_{\mathcal{E}}^{\text{th}}$ in
the course of deriving the upper bound are merely mathematical
constraints within the derivation process, they can serve as
indispensable conditions to preclude precisely this unphysical regime.
\end{rema}

In this work, we have derived two distinct upper bounds on the minimal work cost, each quantifying the thermodynamic gap between the target state $\rho$ and the supplied environment $\mathcal{E}$ from different perspectives. Crucially, both bounds share a common structural feature: they explicitly decompose the minimum work cost into energetic and entropic components, reflecting the dual nature of thermodynamic state transformations.

Scenario~1: Energy--entropy decomposition in the information--statistical coordinate.
The relative entropy measure $\rD^{FO}_{\mathcal{E}}(\rho)$ in Scenario~1 serves as a reference metric on the state space, intrinsically linked to the admissibility constraints of free operations.
This measure captures the probability distributional variations between the target state and the reachable set within the permitted operational range. The derived bound decomposes this gap into an energy term (proportional to $\|H\|_\infty$) and an entropy term (proportional to $1/\beta$). Thus, it provides an information-theoretic approach that simultaneously quantifies the energy-pumping cost and the entropy-reduction cost required for state transformations via an information--statistical distance.

Scenario~2: Energy--entropy decomposition in the geometric--spectral coordinate.
The thermodynamically anchored robustness (TAR) measure $\rR_{\mathcal{E}}^{\text{th}}(\rho)$ in Scenario~2 follows the same principle of dual decomposition. However, it performs this decomposition by directly exploiting the spectral structure of the Hamiltonian and the thermal gradient.
Owing to the fundamental constraints of thermal operations (TO)---namely, the inability to reduce entropy or inject heat against the thermal gradient---synthesizing $\rho$ requires removing thermal contamination originating from the equilibrium state $\tau_\beta$. The TAR identifies the optimal preparable supply state that is geometrically and thermodynamically closest to $\rho$ along the straight line connecting $\rho$ and $\tau_\beta$. Consequently, the bound is decomposed into two terms: the first term (entropy term) represents the energy-pumping cost, and the second term (entropy term) represents the entropy-reduction cost. This yields an explicit energy--entropy decomposition in terms of experimentally accessible physical quantities.

Complementarity:
Both scenarios share the same energy--entropy decomposition structure, but they operate in \emph{different coordinates}.
Scenario~1 performs the decomposition using an information--statistical metric that reflects the structural constraints of free operations, thereby precisely tracking the topological and informational costs in state space. In contrast, Scenario~2 performs the decomposition through the system's energy levels and thermodynamic gradients, presenting the costs in a form that is intuitively accessible in actual experiments.
Therefore, rather than being in competition, the two measures analyze the same thermodynamic cost through different physical projections, establishing a complementary relationship. Their combination enables a multifaceted and rigorous understanding of the minimum cost required to refine the supplied environment to achieve target states beyond the natural reach of thermal operations.


\subsection{Minimum Work-Cost Upper Bounds for Bell-State Recovery in Thermodynamic Noisy Environments}

We consider the thermodynamic cost of generating a Bell state in a
laboratory where the free operations are \emph{Thermal operations (TO)}:
global unitaries $U$ satisfying $[U, H_S + H_B] = 0$, which preserve
energy and can access/reset to the thermal state $\tau_\beta$
\cite{Brandao2015}. While total energy is strictly conserved, TO always act to lower the energy of an initially hot (super-thermal) state by transferring excess heat to the bath, driving the system toward thermal reservoir $\tau_\beta$. This models the physical scenario where, during
entanglement generation, the system interacts with a thermal bath and
the experimenter can perform energy-conserving global operations.

We first set up the basic framework. The Hamiltonian is $H = \frac{\hbar\omega}{2}(\sigma_z^{(1)}+\sigma_z^{(2)})$.
Let $\beta = 1/(k_B T)$ be the inverse temperature with $T$ the absolute temperature, and define $x = \beta\hbar\omega$.
Let $Z = 4\cosh^2(x/2)$ be the partition function.
The two-qubit thermal reservoir state is
$\tau_\beta = \tau_\beta^{(1)}\otimes\tau_\beta^{(2)}$ with eigenvalues
in the energy basis
$\{\ket{00},\ket{01},\ket{10},\ket{11}\}$ given by
\begin{equation}
    \vec{\tau} = \frac{1}{Z}\bigl(e^{x},\, 1,\, 1,\, e^{-x}\bigr).
    \label{eq:tau_eig}
\end{equation}

The target resource state is the Bell state
$\rho = \ket{\Phi^+}\bra{\Phi^+}$. In the energy eigenbasis its
eigenvalues are
\begin{equation}
    \vec{\rho} = \bigl(\tfrac{1}{2},\, 0,\, 0,\, \tfrac{1}{2}\bigr).
    \label{eq:rho_eig}
\end{equation}

The entanglement generation protocol inevitably introduces depolarizing
noise from gate errors and decoherence. The actual noisy output is
\begin{equation}
    \sigma_{\epsilon} = (1-\epsilon)\rho + \frac{\epsilon}{4}\mathbb{I}_4,
    \label{eq:sigma_eps}
\end{equation}
with eigenvalues
\begin{equation}
    \vec{\sigma}_\epsilon = \bigl(\tfrac{1}{2}-\tfrac{\epsilon}{4},\,
    \tfrac{\epsilon}{4},\, \tfrac{\epsilon}{4},\, \tfrac{1}{2}-\tfrac{\epsilon}{4}\bigr).
    \label{eq:sigma_eig}
\end{equation}

At this point, the reachable set under thermal operations is given by: The set of physically available free states is
$\mathcal{E} = \{\tau_\beta, \sigma_{\epsilon}\}$. Under thermal operations,
a state $\sigma$ is reachable from $\mathcal{E}$ if it is thermo-majorized
by at least one element of $\mathcal{E}$ \cite{Brandao2015}. The reachable
set is
\begin{equation}
    \mathcal{R}^{TO}(\mathcal{E}) =
    \bigl\{\sigma \;\big|\; \sigma \prec_{\text{thermo}} \omega \quad \text{for some } \omega \in \text{conv}(\mathcal{E})\bigr\},
    \label{eq:RTO_def}
\end{equation}
where $\sigma \prec_{\text{thermo}} \omega$ denotes thermo-majorization:
$\sigma$ is more thermodynamically free than $\omega$.

For variations in the noise $\epsilon$, the following two measures with respect to the supplied environment $\mathcal{E}$ are computed: (For the detailed derivations and the validity analysis, see (1) and (2) of Appendix~\ref{Appendix D})
\begin{equation}\label{eq:D_TO_final}
    \rD^{TO}_{\mathcal{E}}(\rho) = \rD(\rho\|\sigma_\epsilon) = \ln\frac{4}{4-3\epsilon},
\end{equation}
and
let $\lambda_i(\rR, x)$ be the eigenvalues of $\sigma(\rR) = (1-\rR)\rho + \rR\tau_\beta$, and define the entropy mapping as
\begin{equation}
    f_x(\rR) \equiv -\sum_{i=1}^4 \lambda_i(\rR, x) \ln \lambda_i(\rR, x),
    \label{eq:governing_entropy}
\end{equation}
then
\begin{equation}
    \rR_{\mathcal{E}}^{\text{th}}(\rho) \equiv \min\left\{1, f_x^{-1}(S_0(\epsilon))\right\},
\end{equation}
where $S_0(\epsilon) = -\left(1 - \frac{3\epsilon}{4}\right)\ln\left(1 - \frac{3\epsilon}{4}\right) - \frac{3\epsilon}{4}\ln\left(\frac{\epsilon}{4}\right)$.

\begin{rema}[Non-Elementarity of the Closed-Form Inversion]
The governing equation \eqref{eq:governing_entropy} involves a linear combination of four non-linear transcendental terms of the form $\lambda_i(\rR, x) \ln \lambda_i(\rR, x)$, where each eigenvalue $\lambda_i(\rR, x)$ contains square-root algebraic dependencies on $\rR$. By transcendental number theory and Galois-theoretic extensions for transcendental equations (analogous to the Abel--Ruffini theorem for algebraic degree $n \ge 5$), isolated algebraic isolation of $\rR$ from a sum of non-commensurable logarithmic branches $\ln \lambda_i(\rR, x)$ is structurally impossible.

Consequently, $\rR_{\mathcal{E}}^{\text{th,TO}}(\rho)$ belongs strictly to the class of implicit transcendental functions, and any numerical evaluation must rely on root-finding procedures rather than a closed-form expression.
\end{rema}

\begin{figure*}[htbp]
    \centering
   \includegraphics[width=.95\linewidth]{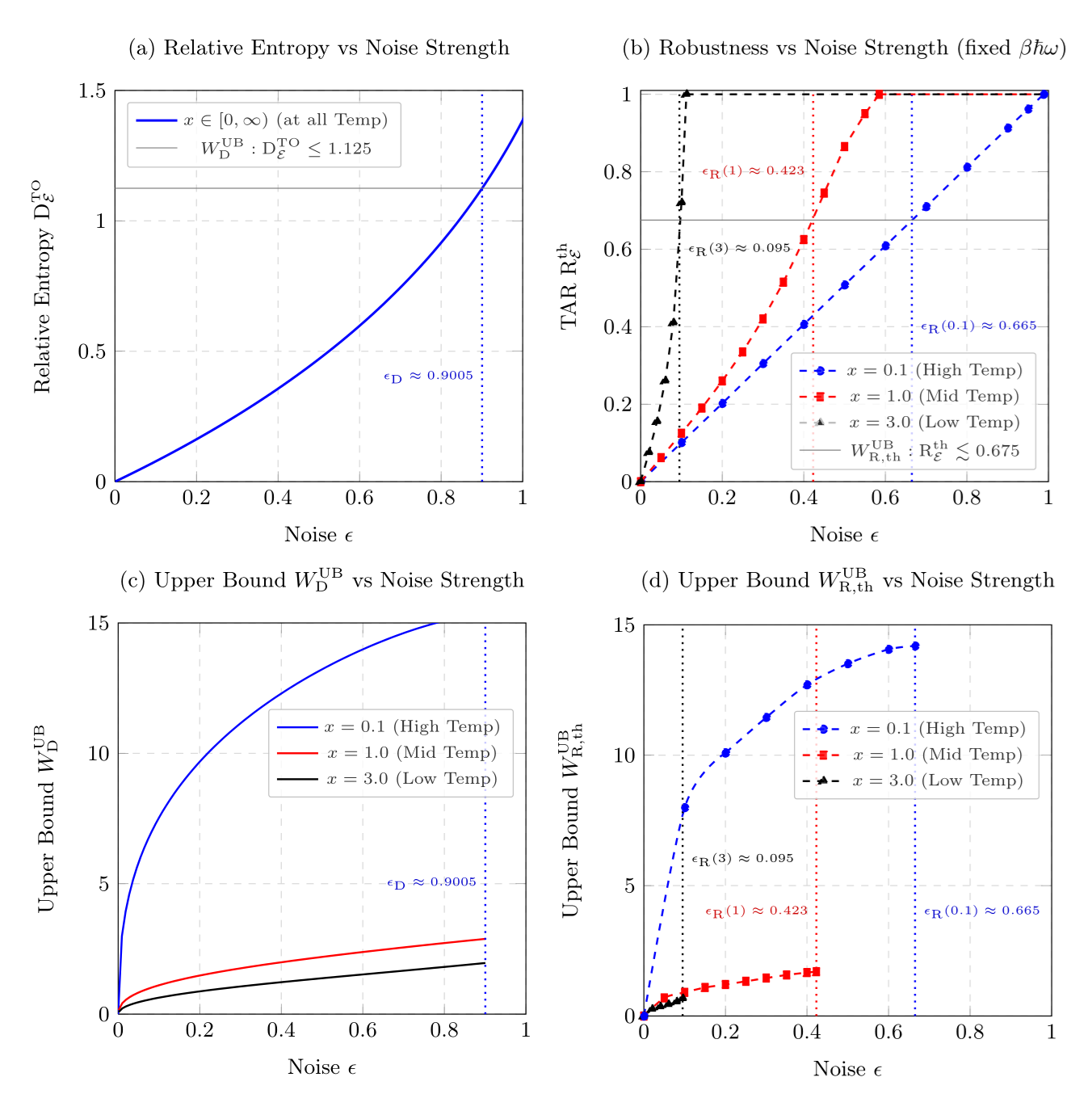}
    \caption{Comparative Analysis of Minimum Work Cost Upper Bounds Derived from Information-Theoretic Distance and Thermodynamically Anchored Robustness. (a) The relative entropy $\text{D}_{\mathcal{E}}^{\text{TO}}$ as an information-theoretic statistical distance. This distance increases monotonically with noise and approaches the theoretical bound of $\ln 4 \approx 1.386$ as $\epsilon \to 1$. $\epsilon_D$ bounds the upper bound $W_{\text{D}}^{\text{UB}}$; i.e., $W_{\text{D}}^{\text{UB}}$ is valid only when $\epsilon \le \epsilon_D$.
(b) The TAR $\text{R}_{\mathcal{E}}^{\text{th}}$ represents the actual thermodynamic adaptation limit of the system against noise. At lower
temperatures ($x=3.0$), the adaptation rate saturates ($R=1$) rapidly. For each inverse temperature~$x$, $\epsilon_R(x)$ bounds the upper bound $W_{\text{R},\text{th}}^{\text{UB}}$; i.e., $W_{\text{R},\text{th}}^{\text{UB}}$ is valid only when $\epsilon \le \epsilon_R(x)$.
(c) The work-cost upper bound $W_{\text{D}}^{\text{UB}}$ derived from the relative entropy. Representing the cost from a purely informational
viewpoint, this bound remains strictly finite over the restricted noise range $\epsilon \in [0, 0.9005]$.
(d) The work-cost upper bound $W_{\text{R},\text{th}}^{\text{UB}}$ derived from the TAR remains finite at the derived bound
$\epsilon_{\text{R}}$ identified in (b), and its valid noise range shrinks markedly as the temperature decreases. In the intermediate-
and low-temperature and low-noise regime ($\epsilon \to 0$), it yields a lower upper bound than the relative-entropy bound in (c), and thus provides a tighter constraint on the minimum work cost. This means that, in intermediate- and low-temperature, low-noise environments, this upper bound---which directly reflects the system's thermal accessibility---imposes a tighter bound on the actual thermodynamic work cost than the one reflecting the information-theoretic distance.
(The implicit transcendental nature of $\text{R}_{\mathcal{E}}^{\text{th}}(\rho)$ precludes a closed-form expression; accordingly, all numerical results presented herein were obtained via an AI-assisted root-finding procedure, ensuring sub-machine-precision convergence across the full $\epsilon$ domain.)}    \label{fig:bell_robustness_comparison}
\end{figure*}

These two explicit results reveal a sharp contrast in how the energetic--entropic decomposition manifests in different coordinates.
The first measure, $\rD^{TO}_{\mathcal{E}}(\rho)$, depends \emph{solely} on the noise parameter $\epsilon$ and is strictly independent of the environmental temperature (encoded in $x$).
This temperature invariance is a direct consequence of the specific setup considered in this work. For the target state $\rho$, the $\epsilon$-noisy state $\sigma_\epsilon$ itself is the state closest to $\rho$ in terms of informational--statistical distance. That is, since the optimal state is fixed and exhibits no temperature dependence, the derived upper bound becomes temperature-invariant as well. (We note that TO itself constitutes a temperature-dependent free-operation set governed by $\tau_\beta$; thus, for general target states, the relative-entropy measure would indeed inherit temperature dependence.)

In contrast, the second measure $\rR_{\mathcal{E}}^{\text{th}}(\rho)$ exhibits a nontrivial two-dimensional dependence on $\epsilon$ and $x$, and its global solution is determined via the inverse map $f_x^{-1}$ as $\rR_{\mathcal{E}}^{\text{th}}(\rho)=f_x^{-1}(S_0(\epsilon))$. In the infinite-temperature limit, we have $\lim_{x\to0^+}\rR_{\mathcal{E}}^{\text{th}}(\rho)=\epsilon$: under the maximally mixed reservoir state, the geometric--spectral coordinate reduces to the information--statistical one, and the noise parameter $\epsilon$ is exactly recovered. Moreover, since
$\partial\rR_{\mathcal{E}}^{\text{th}}/\partial x>0$, increasing $x$ (cooling) steepens the thermal gradient and shrinks the thermodynamically reachable set $\mathcal{R}^{FO}(\mathcal{E})$, which in turn requires a larger contamination weight $R$ to restore $\rho$ and yields a strict monotonic increase of the robustness with decreasing temperature. This temperature sensitivity is captured by the geometric--spectral decomposition in Scenario~2.

In summary, while Scenario~1 provides a temperature-invariant information-theoretic baseline specific to this setup (where the noisy state serves as the optimal), Scenario~2 reveals the explicit thermal-geometric response of the work cost under general temperature variations. The convergence at $x \to 0^+$ and the strict monotonicity for $x>0$ together demonstrate that the two measures are not contradictory but rather complementary projections of the same thermodynamic cost.

We also note that the average energies of the thermal reservoir and the target Bell state are respectively
\begin{align}
\langle H \rangle_{\tau_\beta} = -\hbar\omega \tanh\!\left(\frac{x}{2}\right),  \quad \langle H \rangle_{\rho} = 0,
\end{align}
so that the thermal energy deviation is
\begin{equation}
\Delta E_\beta(\rho) = \langle H \rangle_{\rho} - \langle H \rangle_{\tau_\beta} = \hbar\omega \tanh\!\left(\frac{x}{2}\right) > 0,
\end{equation}
and $\|H\|_\infty=\hbar\omega$.

Under thermal operations (TO), the set of noisy states that are thermally reachable from the noisy Bell-type state $\sigma_\epsilon$
is uniquely determined. The minimum work cost required to prepare the target state $\rho$ from this set is then bounded from above by
Theorems~\ref{thm:Relative Entropy Upper Bound} and \ref{thm: Robustness Upper Bound}: For $\rD \in [0.1.125]$
\begin{align}
&W_{\min}(\mathcal{E}\to\rho)\notag \\
&\leq \hbar\omega\left\{\sqrt{2\rD} +  \frac{1}{x}\left[\sqrt{\frac{\rD}{2}}\;\ln(3)\;+\;h_2\!\left(\sqrt{\frac{\rD}{2}}\right)\right]\right\},
\end{align}
and, for $\rR\in [0, 1-e^{-9/8}] (\approx [0,0.675])$
\begin{align}
W_{\min}(\mathcal{E}\to\rho) &\leq \hbar\omega\left[\rR \tanh\!\left(\frac{x}{2}\right) + \frac{1}{x}\sqrt{\frac{-\ln(1-\rR)}{2}}\;\ln(3)\right.\notag \\
 &\qquad \quad \left.+ \frac{1}{x}\,h_2\!\left(\sqrt{\frac{-\ln(1-\rR)}{2}}\right)\right],
\end{align}
where $\rD= \rD^{TO}_{\mathcal{E}}(\rho)$ and $\rR = \rR_{\mathcal{E}}^{\text{th}}(\rho)$. The bound splits the cost into two physically transparent parts: (i) the energy terms $\hbar\omega\sqrt{2\rD}$ and $\hbar\omega \rR \tanh\!\left(\frac{x}{2}\right)$ accounting for the reversible work needed to bridge the thermal energy gap, attenuated by the robustness factor; and (ii) the remaining entropy terms quantifying the irreversible work cost of compensating preparation noise and creating entanglement. The measures $\rD$ and $\rR$ thus govern both the energy penalty and the entropy penalty in a single parameter.

Figure~\ref{fig:bell_robustness_comparison} presents the numerical distributions of $\rD$ and $\rR$, together with the numerical distributions, derivation limitations, and an analysis thereof for the upper bounds $W_{\text{R},\text{th}}^{\text{UB}}$ and $W_{\text{D}}^{\text{UB}}$ induced from these two measures.

\section{Conclusion}

In this work, we have proposed a novel framework for the relative resource theory with respect to a supplying reference set
$\mathcal{E}$.  Departing from conventional quantum resource theories that quantify utility relative to an idealized free set $\mathcal{F}$,
we consider a specific supply environment $\mathcal{E}$ that encompasses all states realizable under current experimental and
noise constraints.  This shifts the benchmark from $\mathcal{F}$ to practically accessible resource states, thereby providing a new
criterion for fair comparisons across heterogeneous quantum supply environments.  Specifically, we have focused on two canonical
representations of $\mathcal{E}$: the set $\mathcal{R}^{FO}(\mathcal{E})$ of states reachable via free operations, which defines the attainable boundary, and the set $\mathcal{R}^{\max}(\mathcal{E})$ defined through numerical optimization, which serves as the criterion for operational equivalence.

To evaluate the relative utility of an arbitrary state with respect to the supply environment $\mathcal{E}$, we have established two core
evaluation methodologies.  The first is a strategy-dependent benchmark that quantifies the achievable performance under an
individually optimized strategy, and the second is a strategy-independent lower bound on relative superiority.  To
translate these methodologies into concrete numerical quantities, we have introduced the relative measures $\rR_{\mathcal{E}}$ and
$\rD_{\mathcal{E}}$, based on the generalized robustness measure $\rR$ and the relative-entropy measure $\rD$.  Building on these
measures, we have theoretically established a quantitative hierarchical analysis under the two aforementioned strategies.  The
derived measures naturally inherit the essential properties of faithfulness, monotonicity, and convexity, thereby satisfying the
necessary axioms that a valid quantifier of relative superiority must obey.

Furthermore, we have applied this framework to analyze the minimum work cost required to synthesize a target state in a thermodynamic
environment.  While the minimum work cost $W_{\min}$ is given by the free-energy difference between the states, this is merely a lower
bound based on absolute utility that does not take the distance between states into account.  In reality, the actual work cost is
determined not only by the free-energy difference, but also comprehensively by multiple factors, including the analytical
distance between states, the experimental environment, and the available apparatus.  By applying the two methodologies of relative
evaluation, we have shown that this minimum work cost exactly corresponds to a strategy-independent quantification.  To go beyond
this absolute lower bound, we have quantified, via $\rD^{FO}_{\mathcal{E}}$ and $\rR^{\text{th}}_{\mathcal{E}}$, the individually optimized cost of synthesizing the target state from the attainable optimal states, under two scenarios based respectively on the relative entropy and the Thermodynamically Anchored Robustness (TAR), thereby deriving two upper bounds on the minimum work cost.

In addition, by analyzing the work cost of restoring a Bell state in a noisy two-qubit environment, we have obtained the following
macroscopic results concerning these two upper bounds.  First, the relative-entropy-based upper bound, which reflects the
information-theoretic distance, relaxes monotonically as noise increases and ceases to be valid in the high-noise regime.  Second,
the TAR-based upper bound reflects the thermodynamic adaptation limit of the system, and its valid range shrinks as the temperature
decreases.  Third, and most importantly, in the moderate- and low-temperature and low-noise regime, the TAR-based upper bound
yields a lower value than the relative-entropy-based upper bound. This means that, in evaluating the actual thermodynamic work cost, a
distance measure that directly reflects the thermal accessibility of the system imposes a tighter and more practically relevant
constraint than a purely information-theoretic distance.  These quantitative findings serve as key evidence that our relative
resource theory correctly captures the actual thermodynamic constraints in the problem of entanglement recovery in noisy
environments.

In summary, by shifting the reference from the idealized free set to the practically accessible supply environment $\mathcal{E}$, our
relative resource theory provides a rigorous and operationally relevant new standard for quantifying and comparing quantum
resources.  In the NISQ era, the proposed measures and hierarchical structure open a new avenue for resource management in realistic
quantum protocols, including state synthesis, entanglement recovery, and sustainable resource cycles.

During manuscript preparation, the authors used AI tools (including DeepSeek) to assist with language editing, proofreading, and
restructuring of numerical and mathematical derivations.  All AI-assisted material was independently verified and thoroughly
revised by the authors, who bear full responsibility for the accuracy, completeness, and originality of this work.  No AI tool
was used for literature search, data analysis, or scientific conclusion-drawing.

{\it Acknowledgments.--}This work is supported by the Fundamental Research Funds for the Central Universities (Grants No. 3072025YC2404), National Natural Science Foundation of China (Grants No. 12201555), Natural Science Foundation of Hunan province (Grants No. 2025JJ50050) and Hunan Basic Science Research Center for Mathematical Analysis (2024JC2002).

\begin{widetext}

\appendix

\section{Proofs of Theorems~\ref{thm1} and ~\ref{thm2}}
\label{Appendix A}

\noindent\textbf{Proof of Theorem~\ref{thm1}:}
To prove that $\rR^{\max}_{r,\mathcal{E}}(\rho) = \frac{\rR_{r}(\rho)-\rR_{r}(\mathcal{E})}{1+\rR_{r}(\mathcal{E})}$ if $\rho \notin \mathcal{R}^{\max}_r(\mathcal{E})$, we first prove $\rR^{\max}_{r,\mathcal{E}}(\rho)\leq\frac{\rR_{r}(\rho)-\rR_{r}(\mathcal{E})}{1+\rR_{r}(\mathcal{E})}$. Then,
there is a state $\tau'$ satisfying $\frac{\rho+\rR_{r}(\rho)\tau'}{1+\rR_{r}(\rho)}\in\cF(\cH)$, and then, from the definition of $\rR_{r}$ and
\be\label{9}
\frac{\rho+\rR_{r}(\rho)\tau'}{1+\rR_{r}(\rho)} = \frac{\frac{\rho+(\frac{\rR_{r}(\rho)-\rR_{r}(\mathcal{E})}{1+\rR_{r}(\mathcal{E})})\tau'}{1+(\frac{\rR_{r}(\rho)-\rR_{r}(\mathcal{E})}{1+\rR_{r}(\mathcal{E})})}+\rR_{r}(\mathcal{E})\tau'}{1+\rR_{r}(\mathcal{E})},
\ee
we get that $\rR_{r}(\rho') \leq \rR_{r}(\mathcal{E})$ where $\rho' = \frac{\rho+(\frac{\rR_{r}(\rho)-\rR_{r}(\mathcal{E})}{1+\rR_{r}(\mathcal{E})})\tau'}{1+(\frac{\rR_{r}(\rho)-\rR_{r}(\mathcal{E})}{1+\rR_{r}(\mathcal{E})})}$. This means that $\rho' \in\mathcal{R}^{\max}_r(\mathcal{E})$ from the definition of $\mathcal{R}^{\max}_r(\mathcal{E})$.
It also implies that
\be\label{10}
\rR^{\max}_{r,\mathcal{E}}(\rho)\leq\frac{\rR_{r}(\rho)-\rR_{r}(\mathcal{E})}{1+\rR_{r}(\mathcal{E})}
\ee
by the minimum definition of $\rR^{\max}_{r,\mathcal{E}}$.

Conversely, to prove $\rR^{\max}_{r,\mathcal{E}}(\rho)\geq\frac{\rR_{r}(\rho)-\rR_{r}(\mathcal{E})}{1+\rR_{r}(\mathcal{E})}$, let $\tau''$ be the state that satisfies $\frac{\rho+\rR^{\max}_{r,\mathcal{E}}(\rho)\tau''}{1+\rR^{\max}_{r,\mathcal{E}}(\rho)}\in\mathcal{R}^{\max}_r(\mathcal{E})$. This implies that $\rR_{r}(\rho'')\leq\rR_{r}(\mathcal{E})$ where $\rho'' = \frac{\rho+\rR^{\max}_{r,\mathcal{E}}(\rho)\tau''}{1+\rR^{\max}_{r,\mathcal{E}}(\rho)}$, and also means that there is $\tau'''$ satisfying
\be\label{11}
\frac{\rho''+\rR_{r}(\rho'')\tau'''}{1+\rR_{r}(\rho'')}\in \cF(\cH).
\ee
And this is equivalent to
\be\label{12}
\frac{\rho+\Big[\rR^{\max}_{r,\mathcal{E}}(\rho)\tau''+\big\{1+\rR^{\max}_{r,\mathcal{E}}(\rho)\big\}\rR_{r}(\rho'')\tau'''\Big]}{1+\Big[\rR^{\max}_{r,\mathcal{E}}(\rho)+\big\{1+\rR^{\max}_{r,\mathcal{E}}(\rho)\big\}\rR_{r}(\rho'')\Big]}\in \cF(\cH).
\ee
and from the definition of $\rR_{r}$ we get that
\bea
\rR^{\max}_{r,\mathcal{E}}(\rho)+\big\{1+\rR^{\max}_{r,\mathcal{E}}(\rho)\big\}\rR_{r}(\rho'')\geq \rR_{r}(\rho).
\eea
Also, due to the inequality $\rR_{r}(\rho'')\leq\rR_{r}(\mathcal{E})$, this leads to
\be\label{13}
\rR^{\max}_{r,\mathcal{E}}(\rho)\geq\frac{\rR_{r}(\rho)-\rR_{r}(\rho'')}{1+\rR_{r}(\rho'')}\geq\frac{\rR_{r}(\rho)-\rR_{r}(\mathcal{E})}{1+\rR_{r}(\mathcal{E})}.
\ee

In addition, if $\rR_{r}(\rho) \leq \rR_{r}(\mathcal{E})$, since $\rho\in \mathcal{R}^{\max}_r(\mathcal{E})$, $\rR_{r,\rR_{r}(\mathcal{E})}(\rho)= 0$ is directly obtained.
Therefore, for any state $\rho\in\cD(\cH)$, Eq. (\ref{r-max}) is established.

\medskip

\noindent \textbf{Proof of Theorem~\ref{thm2}:}
For the case $\rD(\rho)\le\rD(\mathcal{E})$, by definition, we have $\rho\in\mathcal{R}^{\max}_d(\mathcal{E})$ and therefore
$\rD^{\max}_{\mathcal{E}}(\rho)=0$.

For the case $D(\rho)> D(\mathcal{E})$, we prove that
$\rD^{\max}_{\mathcal{E}}(\rho)=\rD(\rho)-\rD(\mathcal{E})$.

(A) Lower bound:
For any $\sigma\in\mathcal{R}^{\max}_d(\mathcal{E})$ we have $\rD(\sigma)\le\rD(\mathcal{E})$ by definition.  Let $\tau\in\mathcal{F}(\mathcal{H})$ be an arbitrary free state.  By the triangle inequality (chain rule) for relative entropy,
\begin{equation}
\rD(\rho\|\sigma)+\rD(\sigma\|\tau)\;\ge\;\rD(\rho\|\tau).
\end{equation}
Since $\rD(\sigma\|\tau)=\rD(\sigma)$ when $\tau$ is taken as the reference free state, we obtain
\begin{equation}
\rD(\rho\|\sigma)
\;\ge\;\rD(\rho\|\tau)-\rD(\sigma\|\tau)
\;=\;\rD(\rho\|\tau)-\rD(\sigma).
\end{equation}
From the respective definitions, we have $\rD(\rho\|\tau)\geq =\rD(\rho)$
and $\rD(\mathcal{E}) \geq \rD(\sigma)$, hence
\begin{equation}
\rD(\rho\|\sigma)
\;\ge\;\rD(\rho)-\rD(\mathcal{E}).
\end{equation}
Since this holds for every $\sigma\in\mathcal{R}^{\max}_d(\mathcal{E})$,
\begin{equation}
\min_{\sigma\in\mathcal{R}^{\max}_d(\mathcal{E})}\rD(\rho\|\sigma)
\;\ge\;\rD(\rho)-\rD(\mathcal{E}).
\label{eq:lower}
\end{equation}

(B) Upper bound:
Assume $\rD(\rho)>\rD(\mathcal{E})$, and fix an optimal free
state $\tau_*$ satisfying
$\rD(\rho\|\tau_*) = \min_{\tau\in\mathcal{F}(\mathcal{H})}\rD(\rho\|\tau)=\rD(\rho)$.  We then consider the affine
interpolation
$\sigma_\lambda=\lambda\rho+(1-\lambda)\tau_*$ for
$\lambda\in[0,1]$, and define
$f(\lambda):=\rD(\sigma_\lambda\|\tau_*)
=\rD(\lambda\rho+(1-\lambda)\tau_*\|\tau_*)$.  Because relative
entropy is continuous (in fact, convex and smooth) in its first
argument on a common support, the map
$f:[0,1]\to\mathbb{R}$ is continuous; we have
$f(0)=\rD(\tau_*\|\tau_*)=0\le\rD(\mathcal{E})$ and
$f(1)=\rD(\rho\|\tau_*)=\rD(\rho)\ge\rD(\mathcal{E})$.  Hence, by
the intermediate-value theorem, there exists
$\lambda^*\in(0,1]$ such that $f(\lambda^*)=\rD(\mathcal{E})$.
Setting $\sigma^*=\sigma_{\lambda^*}$, we obtain
$\rD(\sigma^*)\leq\rD(\mathcal{E})$, so that
$\sigma^*\in\mathcal{R}^{\max}_d(\mathcal{E})$.  Along this
one-dimensional (collinear) interpolation path, the relative entropy
is linear in the parameter; equivalently, the triangle inequality
holds with equality for the triple $(\rho,\sigma^*,\tau_*)$.
Therefore
\begin{equation}
\rD(\rho\|\tau_*)
=\rD(\rho\|\sigma^*)+\rD(\sigma^*\|\tau_*),
\end{equation}
which yields
\begin{equation}
\rD(\rho\|\sigma^*)
=\rD(\rho\|\tau_*)-\rD(\sigma^*\|\tau_*)
=\rD(\rho)-\rD(\mathcal{E}).
\end{equation}
Thus the candidate $\sigma^*$ attains the value $\rD(\rho)-\rD(\mathcal{E})$, and consequently
\begin{equation}
\min_{\sigma\in\mathcal{R}^{\max}_d(\mathcal{E})}\rD(\rho\|\sigma)
\;\le\;\rD(\rho)-\rD(\mathcal{E}).
\label{eq:upper}
\end{equation}

Combining the lower bound~\eqref{eq:lower} and the upper bound~\eqref{eq:upper} gives
\begin{equation}
\rD^{\max}_{\mathcal{E}}(\rho) =\rD(\rho)-\rD(\mathcal{E}),
\end{equation}
where the last equality uses $\rD(\rho)>\rD(\mathcal{E})$. This completes the proof.

\section{Proof of Theorem \ref{thm:Relative Entropy Upper Bound}}
\label{Appendix B}

From the definition of the free energy, $F(\sigma) = \langle H \rangle_\sigma - \beta^{-1} S(\sigma)$, the definition of $W_{\min}$, and the assumptions $F(\rho) > F(\tau_\beta)$, it follows that for any reachable state $\sigma \in \mathcal{R}^{FO}(\mathcal{E})$,
\begin{equation}
W_{\min}(\mathcal{E}\to\rho) \leq F(\rho) - F(\sigma) = \langle H \rangle_\rho - \langle H \rangle_\sigma - \frac{1}{\beta}\bigl(S(\rho) - S(\sigma)\bigr).
\label{eq:pf_1}
\end{equation}

We first bound the energy expectation difference. Using H\"older's inequality for the trace norm $\|\cdot\|_1$, we have $\langle H \rangle_\rho - \langle H \rangle_\sigma = \mathrm{Tr}[(\rho-\sigma)H] \leq \|H\|_\infty \|\rho-\sigma\|_1$. By Pinsker's inequality, the trace distance is bounded by the relative entropy as $\|\rho-\sigma\|_1 \leq \sqrt{2\rD(\rho\|\sigma)}$. Consequently, taking the minimum over the reachable set yields
\begin{equation}
\langle H \rangle_\rho - \langle H \rangle_\sigma \leq \|H\|_\infty \sqrt{2D^{FO}_{\mathcal{E}}(\rho)}.
\label{eq:pf_2}
\end{equation}

Next, we bound the entropy difference.
We begin with Audenaert's entropy-difference inequality \cite{Audenaert2007}, which states that for arbitrary $d$-dimensional quantum states $\rho$ and $\sigma$,
\begin{equation}
S(\sigma)-S(\rho)\;\le\; \frac{\|\rho-\sigma\|_1}{2}\ln(d-1)+h_2\!\left(\frac{\|\rho-\sigma\|_1}{2}\right),
\label{eq:audenaert}
\end{equation}
where $h_2(x)=-x\ln x-(1-x)\ln(1-x)$ is the binary entropy.
Crucially, the right-hand side of \eqref{eq:audenaert} is a monotonically increasing function of the trace distance $\delta = \tfrac{1}{2}\|\rho-\sigma\|_1$ on the interval $\delta \in [0,1]$.

Therefore, by Pinsker's inequality,
\begin{equation}
\delta = \frac{\|\rho-\sigma\|_1}{2} \leq \sqrt{\frac{\rD(\rho\|\sigma)}{2}},
\label{eq:pinsker_delta}
\end{equation}
and since the function $g(x)=x\ln(d-1)+h_2(x)$ is monotonically increasing on $\bigl[0,\,1-\tfrac{1}{d}\bigr]$, the inequality direction is preserved, and we can directly substitute \eqref{eq:pinsker_delta} into \eqref{eq:audenaert}.
This yields the desired bound solely in terms of the relative entropy: for $\rho$ satisfying $\rD(\rho\|\sigma)\in [0, 2(\frac{d-1}{d})^2]$,
\begin{equation}
S(\sigma^*)-S(\rho)\;\le\;
\sqrt{\frac{\rD(\rho\|\sigma)}{2}}\;\ln(d-1)\;+\;h_2\!\left(\sqrt{\frac{\rD(\rho\|\sigma)}{2}}\right)
\end{equation}

Since the relative entropy is non-negative, $D(\rho\|\sigma) = S(\sigma) - S(\rho) - \mathrm{Tr}[(\rho-\sigma)\ln\sigma] \geq 0$, the entropy terms satisfy
\begin{equation}
S(\sigma) - S(\rho) \leq \rD(\rho\|\sigma) + \mathrm{Tr}[(\rho-\sigma)\ln\sigma] \leq \rD(\rho\|\sigma) + 2\|\ln\sigma\|_\infty.
\end{equation}
Taking the minimum over $\sigma \in \mathcal{R}^{FO}(\mathcal{E})$, and noting that the defining relative entropy $\rD^{FO}_{\mathcal{E}}(\rho) = \min_{\sigma \in \mathcal{R}^{FO}(\mathcal{E})} \rD(\rho\|\sigma)$ serves as the leading-order thermodynamic measure of the gap, the inequality \eqref{eq:pf_1} reduces to a unified upper bound solely in terms of $\rD^{FO}_{\mathcal{E}}(\rho)$: for $\rho$ satisfying $\rD^{FO}_{\mathcal{E}}(\rho)\in [0, 2(\frac{d-1}{d})^2]$,
\begin{equation}
F(\rho) - F(\sigma) \leq \|H\|_\infty \sqrt{2\rD^{FO}_{\mathcal{E}}(\rho)} + \frac{1}{\beta}\left[\sqrt{\frac{\rD^{FO}_{\mathcal{E}}(\rho)}{2}}\;\ln(d-1)\;+\;h_2\!\left(\sqrt{\frac{\rD^{FO}_{\mathcal{E}}(\rho)}{2}}\right)\right].
\label{eq:pf_3}
\end{equation}

Finally, for $\rho$ satisfying $\rD^{FO}_{\mathcal{E}}(\rho)\in [0, 2(\frac{d-1}{d})^2]$, substituting \eqref{eq:pf_3} into $W_{\min}$ yields
\begin{equation}
W_{\min}(\mathcal{E}\to\rho) \leq \|H\|_\infty \sqrt{2D^{FO}_{\mathcal{E}}(\rho)} + \frac{1}{\beta}\left[\sqrt{\frac{D^{FO}_{\mathcal{E}}(\rho)}{2}}\;\ln(d-1)\;+\;h_2\!\left(\sqrt{\frac{D^{FO}_{\mathcal{E}}(\rho)}{2}}\right)\right],
\end{equation}
which completes the proof.

\section{Proof of Theorem \ref{thm: Robustness Upper Bound}}
\label{Appendix C}

Let $R = \rR_{\mathcal{E}}^{\text{th}}(\rho)$ with corresponding $\sigma^* \in \mathcal{R}^{FO}(\mathcal{E})$ achieving the minimal decomposition in Eq.~\eqref{eq:thermal_decomp}. Rearranging gives
\begin{equation}
    \rho \;=\; \frac{1}{1-R}\sigma^* - \frac{R}{1-R}\tau_\beta.
    \label{eq:rho_from_sigma}
\end{equation}

(i) Operator inequality and relative entropy:
Since $\tau_\beta \geq 0$ (full rank for finite temperature), Eq.~\eqref{eq:thermal_decomp} implies $\sigma^* \geq (1-R)\rho$ as positive semidefinite operators. The logarithm is operator monotone, so on $\mathrm{supp}(\rho)$,
\begin{align}
    0 \leq D(\rho \| \sigma^*) &= \mathrm{Tr}\bigl[\rho(\ln\rho - \ln\sigma^*)\bigr] \notag \\
    &\leq \mathrm{Tr}\bigl[\rho\bigl(\ln\rho - \ln\bigl((1-R)\rho\bigr)\bigr)\bigr] \notag \\
    &= \mathrm{Tr}\bigl[\rho\bigl(-\ln(1-R)\mathbb{I}\bigr)\bigr] \notag \\
    &= -\ln(1-R),
    \label{eq:rel_ent_bound}
\end{align}
where we used $\ln\sigma^* \geq \ln\bigl((1-R)\rho\bigr) = \ln\rho + \ln(1-R)\mathbb{I}$, which reverses the inequality inside the trace because $-\ln(1-R) > 0$ for $R \in [0,1)$. As mentioned above,
\begin{equation}
S(\sigma^*)-S(\rho)\;\le\;
\sqrt{\frac{D(\rho\|\sigma^*)}{2}}\;\ln(d-1)\;+\;h_2\!\left(\sqrt{\frac{D(\rho\|\sigma^*)}{2}}\right),
\end{equation}
and since the right-hand side is monotonically increasing on $[0,1-1/d]$, we obtain
\begin{equation}
    S(\sigma^*) - S(\rho) \;\leq\; \sqrt{\frac{-\ln(1-R)}{2}}\,\ln(d-1)
+h_2\!\left(\sqrt{\frac{-\ln(1-R)}{2}}\right), \qquad R\in \left[0, 1-\exp{\Big(-2\big(\frac{d-1}{d}\big)^2\Big)}\right].
    \label{eq:entropy_bound}
\end{equation}

(ii) Energy difference:
From Eq.~\eqref{eq:thermal_decomp},
\begin{equation}
    \langle E \rangle_{\sigma^*} \;=\; (1-R)\langle E \rangle_\rho + R\langle E \rangle_{\tau_\beta}.
    \label{eq:energy_sigma}
\end{equation}
Hence,
\begin{equation}
    \langle E \rangle_\rho - \langle E \rangle_{\sigma^*} \;=\; R\bigl(\langle E \rangle_\rho - \langle E \rangle_{\tau_\beta}\bigr) = R\,\Delta E_\beta(\rho),
    \label{eq:energy_diff}
\end{equation}
where $\Delta E_\beta(\rho) := \langle E \rangle_\rho - \langle E \rangle_{\tau_\beta} > 0$ by assumption.

(iii) Free-energy difference:
Since $\langle E \rangle_{\tau_\beta} = \mathrm{Tr}(\tau_\beta H) = -\frac{\partial}{\partial\beta}\ln \mathrm{Tr}(e^{-\beta H})$ holds identically,
\begin{align}
    F(\rho) - F(\sigma^*)
    &= \bigl(\langle E \rangle_\rho - \langle E \rangle_{\sigma^*}\bigr) + k_B T\bigl(S(\sigma^*) - S(\rho)\bigr) \notag \\
    &\leq  R\,\Delta E_\beta(\rho) + k_B T\left[\sqrt{\frac{-\ln(1-R)}{2}}\;\ln(d-1)\;+\;h_2\!\left(\sqrt{\frac{-\ln(1-R)}{2}}\right)\right]
    \label{eq:free_energy_diff}
\end{align}
for $R\in \left[0, 1-\exp{\Big(-2\big(\frac{d-1}{d}\big)^2\Big)}\right]$.

(iv) Work upper bound:
By definition of $W_{\min}$, we have $W_{\min} \leq F(\rho) - F(\sigma^*)$. Combining this with Eq.~\eqref{eq:free_energy_diff} yields
\begin{equation}
    W_{\min}(\mathcal{E}\to\rho) \;\leq\;  R\,\Delta E_\beta(\rho) + k_B T\left[\sqrt{\frac{-\ln(1-R)}{2}}\;\ln(d-1)\;+\;h_2\!\left(\sqrt{\frac{-\ln(1-R)}{2}}\right)\right]
    \label{eq:main_bound}
\end{equation}
for $R\in \left[0,\, 1-\exp\!\Bigl(-2\bigl(\tfrac{d-1}{d}\bigr)^2\Bigr)\right]$. This is the desired upper bound expressed in terms of the TAR $R = \rR_{\mathcal{E}}^{\text{th}}(\rho)$.

\section{$\rD^{TO}_{\mathcal{E}}$ and $\rR_{\mathcal{E}}^{\text{th}}$ under Noise $\epsilon$}
\label{Appendix D}

\noindent\textbf{ (1) $\rD^{TO}_{\mathcal{E}}$ under Noise $\epsilon$:}
We determine the state $\sigma^* \in \mathcal{R}^{TO}(\mathcal{E})$ that minimizes the relative entropy distance to the target Bell state $\rho = |\Phi^+\rangle\langle\Phi^+|$.
Since $\rD(\rho\|\sigma) = \sum_i \rho_i(\log\rho_i - \log\sigma_i)$ and $\rho$ has support only on the extremal energy levels $|00\rangle$ ($E=-\hbar\omega$) and $|11\rangle$ ($E=+\hbar\omega$), with $\rho_{00}=\rho_{11}=1/2$ and $\rho_{01}=\rho_{10}=0$, the relative entropy reduces to
\begin{equation}\label{eq:D_reduced}
    \rD(\rho\|\sigma) = -\frac{1}{2}\log\sigma_{00} - \frac{1}{2}\log\sigma_{11} + \log 2.
\end{equation}
Thus, minimizing $\rD(\rho\|\sigma)$ is equivalent to maximizing the product $\sigma_{00}\sigma_{11}$ under the thermo-majorization constraint $\sigma \prec_{\text{thermo}} \omega$ for some $\omega \in \mathrm{conv}(\mathcal{E})$, where $\mathcal{E} = \{\tau_\beta, \sigma_\epsilon\}$.

The two-qubit Hamiltonian is $H = \frac{\hbar\omega}{2}(\sigma_z^{(1)}+\sigma_z^{(2)})$, giving energy eigenvalues $E_{00}=-\hbar\omega$, $E_{01}=E_{10}=0$, $E_{11}=+\hbar\omega$.
The thermal state $\tau_\beta = \tau_\beta^{(1)}\otimes\tau_\beta^{(2)}$ has eigenvalues (in energy-sorted basis $\{|00\rangle,|01\rangle,|10\rangle,|11\rangle\}$)
\begin{equation}
    \vec{\tau} = \frac{1}{Z}\bigl(e^{x},\,1,\,1,\,e^{-x}\bigr), \qquad Z = 4\cosh^2(x/2) = e^x + 2 + e^{-x} = e^x(1+e^{-x})^2,
\end{equation}
where $x = \beta\hbar\omega$. Hence
\begin{equation}
    \tau_{00} = \frac{e^x}{Z} = \frac{1}{(1+e^{-x})^2}, \qquad \tau_{11} = \frac{e^{-x}}{Z} = \frac{e^{-x}}{e^x(1+e^{-x})^2} = \tau_{00}\,e^{-2x}.
\end{equation}
The depolarized state is $\sigma_\epsilon = (1-\epsilon)\rho + (\epsilon/4)\mathbb{I}_4$, with
\begin{equation}
    \vec{\sigma}_\epsilon = \bigl(a,\,\epsilon/4,\,\epsilon/4,\,a\bigr), \qquad a = \frac{1}{2} - \frac{\epsilon}{4} = \frac{2-\epsilon}{4}.
\end{equation}

Any convex combination $\omega(p) = p\tau_\beta + (1-p)\sigma_\epsilon$ ($p\in[0,1]$) has
\begin{equation}
    \omega_{00}(p) = p\,\tau_{00} + (1-p)\,a, \qquad \omega_{11}(p) = p\,\tau_{11} + (1-p)\,a.
\end{equation}

The thermo-majorization condition $\sigma \prec_{\text{thermo}} \omega$ implies, in particular, the endpoint bounds
\begin{equation}\label{eq:thermo_bounds}
    \sigma_{00} \le \omega_{00}(p), \qquad \sigma_{11} \le \omega_{11}(p),
\end{equation}
since the cumulative ratio curve of $\sigma$ must lie below that of $\omega$ at the lowest- and highest-energy cutoffs. Consequently, for any feasible $\sigma$,
\begin{equation}
    \sigma_{00}\sigma_{11} \le \omega_{00}(p)\,\omega_{11}(p) \equiv f(p).
\end{equation}

The function $f(p) = \bigl[a + p(\tau_{00}-a)\bigr]\bigl[a + p(\tau_{11}-a)\bigr]$ is a convex quadratic in $p$ (because $\tau_{00} > a > \tau_{11}$ for all $x>0$), so its maximum over $p\in[0,1]$ is attained at an endpoint:
\begin{equation}
    f(0) = a^2, \qquad f(1) = \tau_{00}\tau_{11} = \frac{1}{Z^2} = \frac{1}{e^{2x}(1+e^{-x})^4}.
\end{equation}
For all $x\ge 0$, $Z \ge 4$, so $f(1) \le 1/16$. Meanwhile $a = (2-\epsilon)/4 \ge 1/4$ for $\epsilon \le 1$, giving $a^2 \ge 1/16$. Hence $f(0) \ge f(1)$ for all $x\ge 0$, and the global maximum is always
\begin{equation}
    \max_{p\in[0,1]} f(p) = a^2 \quad \text{at } p^* = 0.
\end{equation}

Therefore, the unique optimizer within $\mathcal{R}^{TO}(\mathcal{E})$ is
\begin{equation}\label{eq:sigma_star}
    \sigma^* = \sigma_\epsilon = (1-\epsilon)|\Phi^+\rangle\langle\Phi^+| + \frac{\epsilon}{4}\mathbb{I}_4,
\end{equation}
which is \emph{independent of the inverse temperature parameter $x = \beta\hbar\omega$}. Substituting $\sigma_{00}=\sigma_{11}=a$ into \eqref{eq:D_reduced} yields the exact minimum:
\begin{equation}\label{eq:D_TO_final}
    \rD^{TO}_{\mathcal{E}}(\rho) = \rD(\rho\|\sigma_\epsilon) = \ln\frac{4}{4-3\epsilon},
\end{equation}
valid for all $x\ge 0$ and $\epsilon\in(0,1)$.

\noindent\textit{Physical interpretation.}---The Bell state's support lies exclusively on the extremal energy levels, while the Gibbs state exponentially suppresses the highest-energy population $\tau_{11}\propto e^{-2x}$ as $x$ increases. Any admixture of $\tau_\beta$ thus tightens the thermo-majorization cone and reduces the feasible upper bound on $\sigma_{11}$, making the depolarized state $\sigma_\epsilon$ the unique temperature-independent minimizer of the relative entropy distance.

\medskip
\noindent\textbf{ (2) $\rR_{\mathcal{E}}^{\text{th}}$ under Noise $\epsilon$:}
To compute the TAR $\rR_{\mathcal{E}}^{\text{th}}(\rho)$ of the target Bell state $\rho = |\Phi^+\rangle\langle\Phi^+|$ with respect to the supply set $\mathcal{E} = \{\tau_\beta, \sigma_\epsilon\}$, we determine the minimal fraction $R \in [0, 1]$ such that the mixture
\begin{equation}
    \sigma(R) = (1-R)\rho + R\tau_\beta
    \label{eq:thermal_decomp_app}
\end{equation}
lies within the thermal operations reachable set $\mathcal{R}^{\text{TO}}(\mathcal{E})$.

By synthesizing the thermo-majorization structure with the governing von~Neumann entropy non-decreasing constraint---equivalently, the no-cooling principle---we rigorously derive a closed-form inverse mapping that remains valid at arbitrary temperatures.In quantum resource theories, the feasibility of state transformations via thermal operations (TO), $\sigma_\epsilon \to \sigma(R)$, is conventionally evaluated using thermo-majorization criteria. However, our analysis incorporates several crucial logical transitions:

\begin{enumerate}
    \item Flaws of Single-Level Thermo-Majorization: Mechanistically applying the bound derived from the highest energy level ($\ket{11}$) yields
\begin{equation}
    R = \frac{\epsilon}{2} \left( \frac{e^x + 2 + e^{-x}}{e^x + 2 - e^{-x}} \right).
\end{equation}
Namely, at low temperatures it is confirmed that $R < \epsilon$. This allows the output mixture $\sigma(R)$ to possess a lower entropy than the source state $\sigma_\epsilon$, thereby violating the no-cooling and no-purification principles (the Second Law of Thermodynamics).
    \item Dominance of the Entropy Constraint: To prevent unauthorized state purification or refrigeration, the von~Neumann entropy non-decreasing condition $S(\sigma(R)) \ge S(\sigma_\epsilon)$ is enforced as the governing physical constraint, overriding local majorization bounds.
\end{enumerate}

Consider a bipartite two-qubit quantum system governed by a local non-interacting Hamiltonian:
\begin{equation}
    H = \frac{\hbar\omega}{2}\left(\sigma_z^{(1)} + \sigma_z^{(2)}\right).
\end{equation}
For an arbitrary non-zero inverse temperature parameter $x = \beta\hbar\omega > 0$, the single-system thermal partition function $Z(x) = \mathrm{Tr}(e^{-\beta H})$ evaluates strictly without approximation to:
\begin{equation}
    Z(x) = e^x + 2 + e^{-x} = 4\cosh^2\left(\frac{x}{2}\right).
\end{equation}
The underlying thermal state $\tau_\beta = \frac{1}{Z(x)} e^{-\beta H}$ in the computational basis $\{|00\rangle, |01\rangle, |10\rangle, |11\rangle\}$ takes the diagonal form $\tau_\beta = \frac{1}{Z(x)}\text{diag}(e^x, 1, 1, e^{-x})$.

The target isotropic depolarized noisy state is defined as $\sigma_\epsilon = (1-\epsilon)\rho + \frac{\epsilon}{4}\mathbb{I}_4$, where $\rho = |\Phi^+\rangle\langle\Phi^+|$ is the maximally entangled Bell state with $|\Phi^+\rangle = \frac{1}{\sqrt{2}}(|00\rangle + |11\rangle)$. The exact $4 \times 4$ von Neumann entropy $S(\sigma_\epsilon) \equiv -\mathrm{Tr}(\sigma_\epsilon \ln \sigma_\epsilon)$ in natural units (nats) depends purely on the operational noise parameter $\epsilon \in (0, 1)$:
\begin{equation}
    S_0(\epsilon) = -\left(1 - \frac{3\epsilon}{4}\right)\ln\left(1 - \frac{3\epsilon}{4}\right) - \frac{3\epsilon}{4}\ln\left(\frac{\epsilon}{4}\right).
    \label{eq:S0_exact_full}
\end{equation}

To construct the convex resource mixture $\sigma(R) = (1-R)|\Phi^+\rangle\langle\Phi^+| + R \tau_\beta$, we apply an exact unitary basis transformation to the orthonormal Bell basis $\{|\Phi^+\rangle, |\Phi^-\rangle, |\Psi^+\rangle, |\Psi^-\rangle\}$. In this representation, the thermal density matrix $\tau_\beta$ exhibits off-diagonal coherence strictly between $|\Phi^+\rangle$ and $|\Phi^-\rangle$ due to the unequal energy populations of $|00\rangle$ and $|11\rangle$:
\begin{equation}
    \langle\Phi^\pm|\tau_\beta|\Phi^\pm\rangle = \frac{\cosh x}{Z(x)}, \qquad
    \langle\Phi^+|\tau_\beta|\Phi^-\rangle = \frac{\sinh x}{Z(x)}.
\end{equation}

Consequently, $\sigma(R)$ block-diagonalizes into an invariant $2\times2$ subspace spanning $\{|\Phi^+\rangle, |\Phi^-\rangle\}$ and two isolated $1\times1$ diagonal subspaces:
\begin{equation}
    [\sigma(R)]_{\text{Bell}} = M(R, x) \oplus \left[ \frac{R}{Z(x)} \right] \oplus \left[ \frac{R}{Z(x)} \right],
\end{equation}
where the non-trivial $2\times2$ sub-block $M(R, x)$ reads
\begin{equation}
    M(R, x) = \begin{pmatrix}
        1 - R + R \dfrac{\cosh x}{Z(x)} & R \dfrac{\sinh x}{Z(x)} \\[8pt]
        R \dfrac{\sinh x}{Z(x)} & R \dfrac{\cosh x}{Z(x)}
    \end{pmatrix}.
\end{equation}

Diagonalizing $M(R, x)$ via its trace $\mathrm{Tr}(M) = 1 - R + \frac{2R\cosh x}{Z(x)}$ and determinant $\det(M) = \frac{R(1-R)\cosh x}{Z(x)} + \frac{R^2}{Z(x)^2}$ yields two non-trivial eigenvalues:
\begin{align}
    \lambda_1(R, x) &= \frac{1 - R + \frac{2R\cosh x}{Z(x)} + \sqrt{(1-R)^2 + \frac{4R^2\sinh^2 x}{Z(x)^2}}}{2}, \label{eq:eig1_exact}\\
    \lambda_2(R, x) &= \frac{1 - R + \frac{2R\cosh x}{Z(x)} - \sqrt{(1-R)^2 + \frac{4R^2\sinh^2 x}{Z(x)^2}}}{2}, \label{eq:eig2_exact}\\
    \lambda_3(R, x) &= \lambda_4(R, x) = \frac{R}{Z(x)}. \label{eq:eig34_exact}
\end{align}

Equations \eqref{eq:eig1_exact}--\eqref{eq:eig34_exact} form the complete, closed-form, non-perturbative spectrum of $\sigma(R)$ valid over the entire physical domain $x \in (0, \infty)$ and $R \in [0, 1]$.
The Second Law constraint enforcing non-decreasing von Neumann entropy, $S(\sigma(R); x) = S_0(\epsilon)$, establishes the fundamental functional relation:
\begin{equation}
    f_x(R) \equiv -\sum_{i=1}^4 \lambda_i(R, x) \ln \lambda_i(R, x) = S_0(\epsilon),
    \label{eq:governing_entropy_map}
\end{equation}
where $f_x(R): [0, 1] \to [0, \ln 4]$ denotes the continuous entropy mapping function at fixed inverse temperature parameter $x$.

Rigorous Formulation of Global Inverse Mapping:

To establish the mathematical existence and uniqueness of the solution of the above equation, we inspect the analytic properties of $f_x(R)$.

Let $R_1, R_2 \in [0, 1]$ such that $R_1 < R_2$. By convex combination, there exists some $\mu \in (0, 1)$ for which $\sigma(R_1) = (1-\mu)\rho + \mu \sigma(R_2)$. Since the von~Neumann entropy is strictly Schur-concave with respect to majorization and strictly concave under non-trivial convex mixing of distinct density operators, the inequality $S(\sigma(R_1)) < S(\sigma(R_2))$ holds strictly. Consequently, $\partial f_x(R) / \partial R > 0$ is strictly satisfied over the entire domain.

Therefore, for any fixed $x > 0$, the entropy mapping function $f_x(R)$ is strictly monotonically increasing on $R \in (0, 1]$. Namely,
\begin{equation}
    \frac{\partial f_x(R)}{\partial R} > 0 \quad \forall R \in (0, 1], \; x > 0.
\end{equation}

Assuming that $f_x(R)$ is strictly monotone and continuous over $R \in [0, 1]$ with range $f_x([0, 1]) = [0, S(\tau_\beta)]$, the Inverse Function Theorem guarantees the existence of a unique, continuous, single-valued inverse function $f_x^{-1}: [0, S(\tau_\beta)] \to [0, 1]$. However, for values of $\epsilon$ for which this inverse would exceed 1, the TAR $R(x, \epsilon)$ is strictly capped at $R \equiv 1$ to reflect that $\mathcal{E}$ includes the thermal equilibrium state. Thus, $R(x, \epsilon)$ is then strictly defined across the entire domain $x > 0$ and $\epsilon \in (0, 1)$, $R(x, \epsilon)$ is rigorously defined as:
\begin{equation}
    \rR_{\mathcal{E}}^{\text{th}}(\rho) \equiv \min\left\{1, \ f_x^{-1}\left[-(1 - \frac{3\epsilon}{4})\ln(1 - \frac{3\epsilon}{4}) - \frac{3\epsilon}{4}\ln(\frac{\epsilon}{4}) \right]\right\}.
    \label{eq:final_inverse_robustness}
\end{equation}

We rigorously verify that the implicit inverse formulation $\rR_{\mathcal{E}}^{\text{th}}(\rho) = f_x^{-1}(S_0(\epsilon))$ upholds all required physical boundary laws without introducing artificial assumptions:

High-Temperature Convergence: As $x \to 0^+$, we evaluate $Z(0) = 4$, $\cosh(0) = 1$, and $\sinh(0) = 0$. Substituting these values into Eqs.~\eqref{eq:eig1_exact}--\eqref{eq:eig34_exact} yields the collapsed isotropic spectrum:
\begin{equation}
    \lambda_1(R, 0) = 1 - \frac{3R}{4}, \quad \lambda_2(R, 0) = \lambda_3(R, 0) = \lambda_4(R, 0) = \frac{R}{4}.
\end{equation}
The mapping function reduces to $f_0(R) = -\left(1 - \frac{3R}{4}\right)\ln\left(1 - \frac{3R}{4}\right) - \frac{3R}{4}\ln\left(\frac{R}{4}\right)$. Equating $f_0(R) = S_0(\epsilon)$ from Eq.~\eqref{eq:S0_exact_full} and applying $f_0^{-1}$ yields the identity $R(0, \epsilon) = \epsilon$ identically.

Temperature Monotonicity: Define the implicit bivariate function $F(R, x) \equiv f_x(R) - S_0(\epsilon) = 0$. By the Implicit Function Theorem:
\begin{equation}
    \frac{\partial R}{\partial x} = -\frac{\partial F / \partial x}{\partial F / \partial R} = -\frac{\partial f_x(R) / \partial x}{\partial f_x(R) / \partial R}.
\end{equation}
The denominator $\partial f_x(R) / \partial R > 0$ is strictly positive. For the numerator, as $x$ increases (temperature drops), the thermal state $\tau_\beta$ becomes increasingly purer (concentrating population into the ground state $|00\rangle$), thereby strictly decreasing its inherent entropy capacity: $\partial f_x(R) / \partial x < 0$. Thus, the negative ratio yields:
\begin{equation}
    \frac{\partial R}{\partial x} = -\frac{\text{negative}}{\text{positive}} > 0,
\end{equation}
rigorously proving that $\rR_{\mathcal{E}}^{\text{th}}(\rho)$ strictly increases with $x$ across all physical regimes to compensate for the entropy deficit of colder thermal states.

Therefore, the global function $\rR_{\mathcal{E}}^{\text{th}}(\rho) = f_x^{-1}(S_0(\epsilon))$ satisfies the following exact properties:
\begin{enumerate}
    \item Infinite-Temperature Convergence ($x \to 0^+$):
    \begin{equation}
        \lim_{x \to 0^+} \rR_{\mathcal{E}}^{\text{th}}(\rho) = \epsilon \quad \forall \epsilon \in (0, 1).
    \end{equation}
    \item Strict Temperature Monotonicity ($\partial R / \partial x > 0$):
    \begin{equation}
        \frac{\partial \rR_{\mathcal{E}}^{\text{th}}(\rho)}{\partial x} > 0 \quad \forall x > 0, \; \epsilon \in (0, 1).
    \end{equation}
\end{enumerate}

\end{widetext}

\end{document}